\documentclass[twocolumn]{IEEEtran}

\usepackage{amsmath,graphics,amssymb,graphicx}
\usepackage{subfigure}
\usepackage{url}
\usepackage{color,soul}
\usepackage{verbatim}
\usepackage{tabularx}
\usepackage{cite}
\usepackage{multirow}
\usepackage{subfigure}
\usepackage{booktabs}
\usepackage{amsmath,amssymb}
\usepackage{bm}
\usepackage{algorithm}
\usepackage{algpseudocode}
\usepackage{amsthm}
\usepackage{booktabs}
\usepackage[table]{xcolor}

\newtheorem{theorem}{Theorem}

\newtheorem{proposition}{Proposition}

\begin{document}
\setlength{\abovedisplayskip}{5pt}
\setlength{\belowdisplayskip}{5pt}
\title{{Channel Geometry Preserving Generative Models for CSI Feedback in MU-MIMO}}

\author{\IEEEauthorblockN{Juseong~Park, Taekyun~Lee, Foad~Sohrabi, and Jeffrey~G.~Andrews} \\ 

\thanks{This work has been supported by Nokia Bell Labs and in part by the National Science Foundation RINGS program grant CNS-2148141.}
\thanks{Juseong~Park, Taekyun~Lee and Jeffrey~G.~Andrews are with 6G@UT in the Wireless Networking and Communications Group, The University of Texas at Austin, Austin, TX 78712, USA (email: juseong.park@utexas.edu; taekyun@utexas.edu; jandrews@ece.utexas.edu).}
\thanks{Foad Sohrabi is with Radio Systems Research, Nokia Bell Labs, Murray Hill, NJ 07974 USA (e-mail: foad.sohrabi@nokia-bell-labs.com).}
}
\maketitle

\begin{abstract}
Under limited feedback, channel state information (CSI) reconstruction for multiuser multiple-input multiple-output (MU-MIMO) precoding is challenging, since the precoder should provide not only beamforming gain, but also robust suppression of inter-user interference. 
This paper revisits this classic problem by developing powerful decompression techniques at the base station (BS) that harness modern deep generative models.  
We propose two novel BS-side flow-matching generative CSI decoders that progressively transform either a simple prior or an initial CSI estimate into a reconstruction consistent with the feedback-conditioned channel distribution.
We further show theoretically that conventional minimum mean-squared-error (MMSE)-based reconstructions of CSI often result in centroid-like compromises that fail to preserve the posterior geometry needed for inter-user interference suppression.  
In other words, MU-MIMO precoding based on MSE-oriented CSI reconstructions can be suboptimal, since such reconstructions frequently fail to maintain user orthogonality.
Numerical results in FR3 spectrum show that the proposed flow-based methods consistently outperform MSE-based baselines in downlink sum-rate, with the advantage especially pronounced in interference-limited and spatially dense regimes. 
These results suggest that posterior-guided flow reconstruction is better aligned with MU-MIMO precoding than traditional MMSE-oriented CSI feedback, since it better preserves the channel geometry needed for user separation.
\end{abstract}


\section{Introduction} \label{sec:introduction}
Multiuser multiple-input multiple-output (MU-MIMO) promises many-fold gains in downlink spectral efficiency for cellular systems relative to single user (SU) MIMO, enabled by the large asymmetry between the large antenna array at the base station (BS) and the much smaller one at the UE \cite{Goldsmith03}.  These gains have remained doggedly difficult to realize in practice, for several reasons \cite{Andrews25}.  One important reason is that suppressing inter-user interference requires accurate channel state information (CSI) at the BS, implying a large number of feedback bits \cite{Jindal06}.   Artificial intelligence (AI)-enhanced CSI feedback has emerged as a promising direction, since it can provide better MIMO channel reconstruction with fewer bits \cite{Guo22}.  AI-aided MIMO channel compression and reconstruction has attracted growing interest in standardization \cite{Xingqin25}.  However, most existing approaches still treat CSI feedback mainly as a problem of accurate channel reconstruction. In the context of downlink MU-MIMO precoding, the key issue is not the per-user CSI reconstruction fidelity, but instead, the ultimate performance of the MU-MIMO precoder.  Although related, a determining factor for MU-MIMO precoding efficacy is whether the CSI feedback and reconstruction preserves the relative channel geometry, since that geometry governs interference suppression across users. When this geometry is distorted, interference suppression degrades and the achievable sum rate is quickly reduced. This suggests that effective CSI reconstruction depends more on preserving precoding-relevant geometry than on minimizing conventional reconstruction error in a mean squared error (MSE) sense.

\subsection{Related Work}
\label{subsec:related_work}

Prior to the recent consideration of AI-based CSI \mbox{feedback} methods, finite-rate feedback---often referred to as limited feedback---has been extensively studied in MIMO systems. Broad overviews of this literature can be found in \cite{Love08} and Section~5.10 of \cite{Heath18}. Representative works include Grassmannian codebook design for single-user beamforming \cite{Love03}, finite-rate CSI feedback for zero-forcing precoding in the multiuser setting \cite{Jindal06}, and block diagonalization with multi-antenna users through channel subspace quantization based on chordal distance \cite{Ravindran08}.

More recently, AI-based CSI feedback has emerged as an active research direction with some promising results. A notable early work is CsiNet \cite{Wen18}, an autoencoder framework that compresses the channel matrix into a low-dimensional real-valued latent vector using a convolutional neural network. 
Although the latent representation is lower-dimensional than the original channel vector, it remains real-valued and must therefore be quantized for finite-rate feedback. Common approaches include uniform quantization \cite{Lu20_uniform}, $\mu$-law quantization \cite{Guo20}, and learned quantization \cite{Yin25}, which are largely scalar quantization schemes applied element-wise to latent representations. Beyond scalar quantization, vector quantization has also been studied as a more efficient alternative that exploits dependencies among latent elements \cite{Rizzello23}.
End-to-end MIMO architectures have also studied CSI feedback through the joint optimization of pilot transmission, feedback, and precoding \cite{Sohrabi21, Jang22, Park25, Zhu26, Carpi26, Chen26}. In these works, user-side neural networks generate feedback bits or custom pilots (possibly analog), and then the entire system is trained to maximize the sum rate. As a result, the feedback is not aimed so much at accurate channel reconstruction, but instead at its usefulness in computing an effective precoder.

Generative diffusion models have recently been explored across a broad range of wireless communication tasks \cite{Fan26}. Diffusion-based channel priors have also been used for MIMO channel estimation \cite{Fesl24,Chen25,Fan26}. More directly related to CSI feedback, diffusion models have been applied to side-information-aided CSI compression and variable-rate joint source-channel coding (JSCC) of MIMO CSI \cite{Kim25,Ankireddy26}. While these CSI-feedback-oriented works show that diffusion models can improve CSI reconstruction quality in challenging settings, they still evaluate performance primarily through reconstruction-oriented criteria such as MSE. Thus, they do not directly address how the learned feedback representation affects MU-MIMO precoding quality or the resulting sum-rate performance.

From a deployment perspective, recent work has explored BS-side AI-based CSI feedback architectures, where the learned model is placed at the BS while the UE retains a lightweight feedback mechanism \cite{Guo25}. Despite substantial progress in CSI feedback, the relationship between finite-rate CSI feedback and MU-MIMO precoding performance remains insufficiently understood, especially when the key issue is not channel reconstruction fidelity itself but the preservation of channel geometry relevant to interference suppression.

\subsection{Contributions}
\label{subsec:contributions}

This paper studies finite-rate CSI feedback for MU-MIMO precoding from a BS-side decoding perspective. Representing continuous-valued CSI with finite feedback bits inevitably introduces reconstruction errors, and the key issue is how the resulting channel posterior should be exploited for MU-MIMO precoding.
Thus, we develop BS-side generative CSI decoders together with a theoretical analysis showing that reconstruction methods that better preserve channel geometry are more effective for MU-MIMO precoding. The main contributions of this paper are summarized as follows.

\subsubsection{Flow-matching generative CSI decoders}
We propose two BS-side flow-matching generative decoders for CSI feedback in MU-MIMO precoding, assuming that the UE-side encoder and feedback bits are given.
Flow matching learns a continuous source-to-target flow by matching a neural vector field to transport directions; in our setting, the flow is conditioned on the feedback bits and targets the corresponding channel distribution.
The first is a \emph{flow-based CSI refiner}, which initializes from a deterministic front-end estimate and performs multi-step posterior-guided generative decoding to produce a refined CSI estimate. 
The second is a \emph{direct conditional flow decoder}, which reconstructs CSI directly from the feedback bits by transporting a simple prior toward a feedback-conditioned channel distribution. 
The use of flow matching is motivated by its flexibility in specifying the source distribution: unlike diffusion models \cite{Ho20}, which generate samples by reversing a Gaussian-noise process, flow matching \cite{Lipman23} can connect an arbitrary source distribution to a desired target distribution. 
This flexibility is particularly useful in our setting, where the refiner starts from an MSE-based initial estimate for guided refinement, whereas the direct decoder starts from Gaussian noise and relies on conditioning to model the posterior channel distribution. 
The two models therefore provide complementary operating points.

\subsubsection{Posterior directional geometry analysis for MU-MIMO precoding}
We develop a theoretical analysis of posterior directional geometry for MU-MIMO precoding under finite-rate CSI feedback. By decomposing the signal-to-interference-plus-noise ratio (SINR) into the desired-signal and interference terms, we characterize how posterior geometry affects each component differently. For the desired-signal term, we show that the relevant local criterion is a conditional expected chordal distortion, whose posterior-optimal direction is given by the principal eigenvector of the conditional second-moment matrix, thereby clarifying when the MSE-induced conditional mean becomes suboptimal. 
For the interference term, we develop an interference characterization and a finite-mixture surrogate for comparing posterior sampling with conditional-mean averaging. The comparison is governed by three competing effects: (i) the geometry-collapsing effect of conditional-mean averaging, (ii) the residual overlap left under posterior sampling, and (iii) the separation of the resulting directions. These results clarify when MSE-oriented channel reconstruction becomes unfavorable for MU-MIMO precoding and yield a sufficient condition under which posterior sampling reduces interference relative to conditional-mean averaging.

\subsubsection{Empirical evidence beyond MSE-oriented CSI feedback}
We provide extensive numerical experiments for MU-MIMO downlink precoding on 7~GHz 3GPP UMi channels generated by QuaDRiGa. The experiments show that MSE-oriented CSI feedback alone is not sufficient for high-quality MU-MIMO precoding, since lower normalized MSE (NMSE) does not necessarily imply better sum-rate. The proposed flow-based decoders consistently outperform MSE-oriented deterministic baselines, with gains becoming more pronounced in interference-limited and spatially dense regimes. In an eight-user setting with 48 feedback bits per user at SNR $=20$~dB, they achieve up to $2.4$ times the sum-rate of MSE-based deterministic baselines while reaching about $89\%$ of Full CSI.
Aggregate desired-signal and interference diagnostics, together with DFT-domain and posterior feedback-cell analyses, further substantiate this conclusion. 
Our findings are analogous to recurring observations in computer vision that MSE-based objectives often produce blurry or otherwise suboptimal reconstructions from a perceptual-quality standpoint \cite{Mathieu16, Pathak16, Blau18}. Here, the relevant structure is not visual fidelity, but the spatial channel geometry which ultimately governs MU-MIMO precoding performance.

The rest of this paper is organized as follows. Section~\ref{sec:system_model} introduces the system model and formulates finite-rate CSI feedback for MU-MIMO precoding. Section~\ref{sec:proposed_method} presents the proposed flow-based CSI refiner and direct conditional flow decoder. Section~\ref{sec:posterior_directional_geometry} develops the posterior directional-geometry analysis under limited CSI feedback and studies how posterior sampling and MSE-oriented reconstruction interact with ZF interference geometry. Section~\ref{sec:numerical_results} provides numerical results and further diagnostic analyses. Finally, Section~\ref{sec:conclusion} concludes the paper.

\begin{figure*}[tbp]
    \centering
    \includegraphics[width=0.8\textwidth]{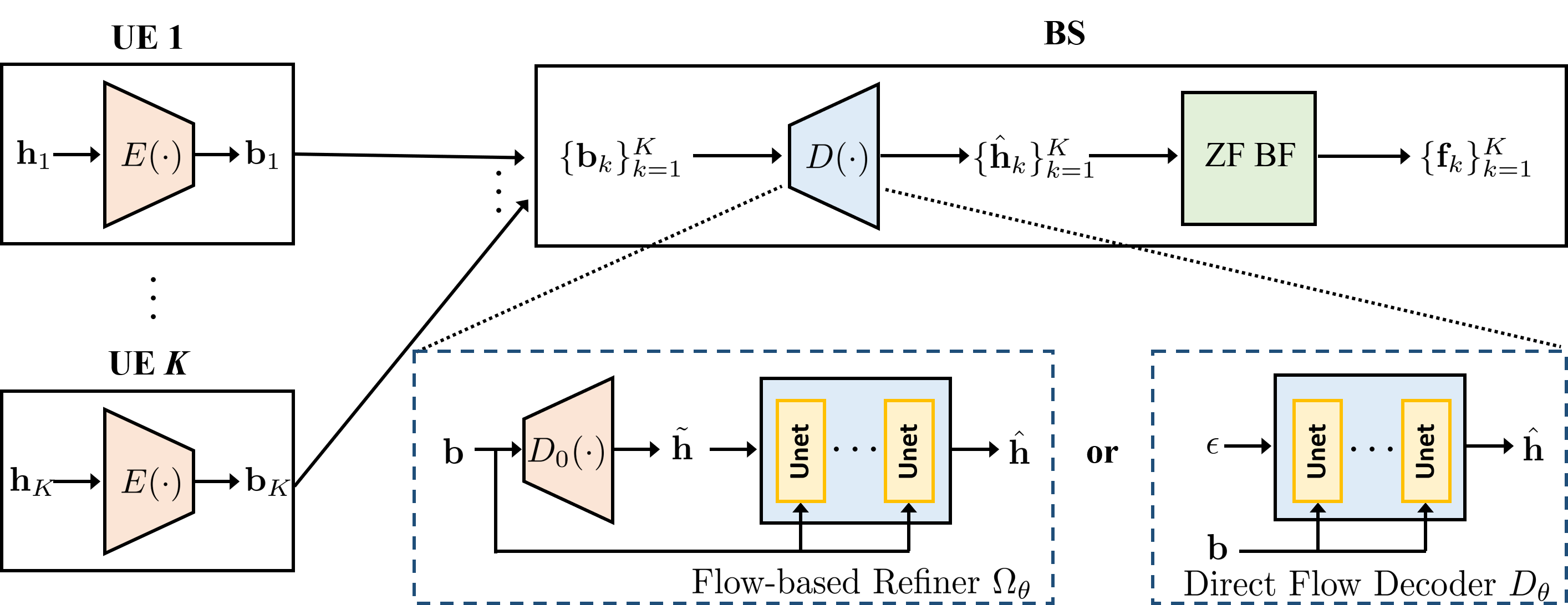}
    \caption{The proposed CSI feedback methods for MU-MIMO precoding.}
    \label{figure:proposed_method}
\end{figure*}

\section{System Model}
\label{sec:system_model}

Consider a downlink MU-MIMO system in which a BS equipped with $N$ antennas simultaneously serves $K$ single-antenna UEs, where $K \leq N$. Let $\mathbf{h}_k \in \mathbb{C}^{N}$ denote the instantaneous downlink channel vector of UE $k$. Each UE is assumed to acquire its own CSI through downlink training and to feed back a finite-rate CSI representation to the BS.

Specifically, UE $k$ applies a $B$-bit feedback encoder $E(\cdot)$ to generate the feedback bits
\begin{equation*}
\mathbf{b}_k = E(\mathbf{h}_k) 
 \in \{0,1\}^B,
\end{equation*}
where $B$ denotes the feedback budget per user. Upon receiving $\mathbf{b}_k$, the BS reconstructs the channel as
\begin{equation*}
\hat{\mathbf{h}}_k = D(\mathbf{b}_k),
\end{equation*}
where $D(\cdot)$ denotes the BS-side CSI decoder. The specific form of $D(\cdot)$ depends on the reconstruction architecture and will be detailed in Section~\ref{sec:proposed_method}.

Let $\mathbf{H} \triangleq [\mathbf{h}_1,\ldots,\mathbf{h}_K]^H \in \mathbb{C}^{K\times N}$ denote the true channel matrix and $\hat{\mathbf{H}}$ its reconstructed version at the BS. 
The BS employs a normalized ZF precoder based on $\hat{\mathbf{H}}$. The corresponding unnormalized precoder is
\begin{equation*}
\tilde{\mathbf{F}} = \hat{\mathbf{H}}^H \left(\hat{\mathbf{H}}\hat{\mathbf{H}}^H\right)^{-1},
\end{equation*}
where the $k$th column of $\tilde{\mathbf{F}}$ is denoted by $\tilde{\mathbf{f}}_k$. To satisfy a total transmit-power constraint $P$ with equal power allocation across users, each beamformer is normalized as
$\mathbf{f}_{k} = \sqrt{\frac{P}{K}}~
{\tilde{\mathbf{f}}_{k}}
/{\|\tilde{\mathbf{f}}_{k}\|_2}.$
Thus, the final precoding matrix is
\begin{equation*}
\mathbf{F} = [\mathbf{f}_1,\ldots,\mathbf{f}_K] \in \mathbb{C}^{N \times K}.
\end{equation*}

For the transmitted symbol vector $\mathbf{s} \in \mathbb{C}^{K}$ satisfying
$\mathbb{E}[\mathbf{s}\mathbf{s}^H] = \mathbf{I}$, the BS transmit signal is given by
$\mathbf{x} = \mathbf{F}\mathbf{s}$. The received signal at UE $k$ is then
\begin{equation*}
r_k= \mathbf{h}_k^H \mathbf{f}_k s_k +
\sum_{n\neq k}\mathbf{h}_k^H \mathbf{f}_n s_n + n_k,
\end{equation*}
where $n_k \sim \mathcal{CN}(0,\sigma^2)$ denotes additive white Gaussian noise. The corresponding achievable rate is
\begin{equation}
R_k =
\log_2\!\left(1 + \frac{|\mathbf{h}_k^H \mathbf{f}_k|^2}{\sigma^2 + \sum_{n\neq k} |\mathbf{h}_k^H \mathbf{f}_n|^2} \right).
\label{eqn:sinr}
\end{equation}

Under perfect CSI, ZF completely removes inter-user interference. Under finite-rate CSI feedback, however, the precoder is designed from $\hat{\mathbf{H}}$ rather than $\mathbf{H}$, and residual inter-user interference generally remains. Thus, the downlink sum-rate performance depends critically on how the received feedback bits are converted into BS-side CSI.

In this paper, we consider a practical architecture in which the UE-side encoder $E(\cdot)$ and the corresponding BS-side initial decoder $D_0(\cdot)$ are assumed to be given and fixed; we refer to this pair as the \emph{front end}.
The goal is then to determine how the BS-side decoder $D(\cdot)$ should reconstruct CSI for ZF precoding from the information provided by the front end, either the initial estimate $D_0(\mathbf{b}_k)$ or the feedback bits $\mathbf{b}_k$. We evaluate $D(\cdot)$ through the resulting expected downlink sum rate
\begin{equation*}
R_{\mathrm{sum}}=\mathbb{E}\left[\sum_{k=1}^{K} R_k\right],
\end{equation*}
where $\hat{\mathbf{h}}_k$ denotes the final reconstruction produced by $D(\cdot)$ and each $R_k$ is computed using the ZF beams obtained from $\hat{\mathbf{H}}$.

Since the mapping $\mathbf{b}_k = E(\mathbf{h}_k)$ is many-to-one, each feedback vector generally corresponds to multiple plausible channel realizations rather than to a unique channel vector. Therefore, under finite-rate feedback, exact channel recovery is generally impossible, and the key issue is how the decoder exploits the resulting channel posterior to produce a channel estimate that preserves the spatial geometry needed for effective ZF precoding.

\section{Proposed Method}
\label{sec:proposed_method}

In this section, we present two flow-based CSI reconstruction methods: \emph{flow-based CSI refiner} and \emph{direct conditional flow decoder}, as illustrated in Fig.~\ref{figure:proposed_method}. 

To follow the flow-matching convention in \cite{Lipman23}, we use an equivalent real-valued representation of the complex channel obtained by stacking its real and imaginary parts. Thus, $\mathbf{h} \in \mathbb{C}^{N}$ is represented in $\mathbb{R}^{2N}$, or equivalently as a $2 \times N$ tensor in implementation. With a slight abuse of notation, we continue to denote this real-valued representation by $\mathbf{h}$ throughout this section.
To simplify notation, we drop the user index in this section.

\subsection{Overview of Flow Matching}
\label{sec:fm_overview}

We briefly review the flow-matching framework underlying the proposed decoders. Flow matching learns a time-dependent neural vector field that transports samples from a source distribution to a target data distribution \cite{Lipman23}. This transport is described by an ordinary differential equation (ODE) parameterized by a flow time variable $t\in[0,1]$, which is an internal generative variable rather than the physical time index of the wireless channel. 
Flow matching directly learns the transport dynamics, avoiding the explicit noise-corruption and reversal formulation commonly used in diffusion models.

Formally, flow matching specifies a reference probability path between a source distribution and a target distribution, and trains a time-dependent neural vector field $\boldsymbol{\phi}_{\theta}(\mathbf{x},t)$ to match the corresponding target vector field $\mathbf{w}(\mathbf{x},t)$. 
In the optimal-transport view of flow matching, a natural choice for the probability path is the straight-line interpolation \cite{Lipman23},
\begin{equation*}
\mathbf{x}(t)=(1-t)\mathbf{x}(0)+t\mathbf{x}(1),
\end{equation*}
between a source sample $\mathbf{x}(0)\sim q_0$ and a target sample $\mathbf{x}(1)\sim q_1$, which yields the closed-form target direction
\begin{equation*}
\mathbf{w}(\mathbf{x}(t),t) \triangleq
\frac{d\mathbf{x}(t)}{dt} =
\mathbf{x}(1)-\mathbf{x}(0).
\end{equation*}
The parameter $\theta$ is learned by sampling points along this path and fitting $\boldsymbol{\phi}_{\theta}(\mathbf{x}(t),t)$ to the corresponding target transport direction $\mathbf{w}(\mathbf{x}(t),t)$.

At generation time, the learned vector field governs the transport through
\begin{equation*}
\frac{d\mathbf{x}(t)}{dt} =
\boldsymbol{\phi}_{\theta}(\mathbf{x}(t),t), \qquad t\in[0,1],
\end{equation*}
with initialization $\mathbf{x}(0)\sim q_0$ and output given by the terminal state $\mathbf{x}(1)$. The transport is implemented by numerically integrating this ODE from $t=0$ to $t=1$ \cite{Lipman23}.

This formulation offers two practical advantages in our setting. First, the source distribution can be adapted to the reconstruction task, allowing generation to start either from an existing front-end channel estimate or from a simple prior. Second, since generation follows a learned ODE transport, sampling can be more efficient than diffusion models.

In addition, in the proposed methods, the transport in both cases is conditioned on the received feedback bits and targets the corresponding conditional channel distribution. The difference lies in the source distribution: the flow-based CSI refiner starts from a source distribution centered around the front-end reconstruction, whereas the direct conditional flow decoder starts from a simple Gaussian prior.

\subsection{Flow-Based CSI Refiner}
\label{sec:flowdec_refiner}

We first consider a two-stage BS-side decoding architecture built on the fixed front end introduced in Section~\ref{sec:system_model}. Given feedback bits $\mathbf{b}$, the BS obtains a deterministic front-end estimate $\tilde{\mathbf{h}} = D_0(\mathbf{b})$, where $D_0(\cdot)$ denotes the BS-side initial decoder associated with the fixed front end. As illustrated in Fig.~\ref{figure:proposed_method}, the flow-based CSI refiner then produces
\begin{equation*}
\hat{\mathbf{h}} = D(\mathbf{b})
= \Omega_{\theta}(\tilde{\mathbf{h}},\mathbf{b})
= \Omega_{\theta}(D_0(\mathbf{b}),\mathbf{b}),
\end{equation*}
where $\Omega_{\theta}(\cdot,\cdot)$ denotes the refinement module conditioned on the feedback bits. Thus, $D(\cdot)$ refines the deterministic front-end estimate using posterior information implied by $\mathbf{b}$.

The refiner specializes the flow-matching framework in Section~\ref{sec:fm_overview} to feedback-conditioned CSI reconstruction. 
It operates as a deterministic refinement of the front-end reconstruction at inference, where the source corresponds to the point mass at $\tilde{\mathbf{h}}=D_0(\mathbf{b})$. 
During training, however, we use a locally perturbed version of this source, following \cite{Welker25}, to regularize the flow and learn a locally stable refinement field around the front-end estimate.
Given $\mathbf{b}$, the target distribution is $q_1(\mathbf{h}\mid\mathbf{b})=p(\mathbf{h}\mid\mathbf{b})$, and the training source distribution is chosen as
\begin{equation*}
q_0(\mathbf{h}\mid\mathbf{b}) =
\mathcal{N}\left(\tilde{\mathbf{h}},\sigma_0^2\mathbf{I}\right),
\end{equation*}
where $\sigma_0>0$ controls the perturbation level around the front-end estimate.
The refiner therefore learns a feedback-conditioned vector field
\begin{equation}
\frac{d\mathbf{h}(t)}{dt} =
\boldsymbol{\phi}_{\theta}(\mathbf{h}(t),\mathbf{b},t), \qquad t\in[0,1],
\label{eq:refiner_ode}
\end{equation}
that transports samples from the front-end-centered training source distribution toward the conditional channel distribution.

To train the refiner, for each training pair $(\mathbf{h},\mathbf{b})$, we sample $t\sim\mathcal{U}(0,1)$ and $\boldsymbol{\epsilon}\sim\mathcal{N}(\mathbf{0},\mathbf{I})$, and define
\begin{equation}
\mathbf{h}(t) = (1-t)\bigl(\tilde{\mathbf{h}}+\sigma_0\boldsymbol{\epsilon}\bigr) + t\mathbf{h}.
\label{eq:refiner_interpolation}
\end{equation}

The corresponding conditional path distribution is
\begin{equation*}
p(\mathbf{h}(t)\mid \mathbf{h},\mathbf{b}) =
\mathcal{N}\!\left( \tilde{\mathbf{h}} + t(\mathbf{h}-\tilde{\mathbf{h}}), (1-t)^2\sigma_0^2\mathbf{I} \right),
\end{equation*}
whose mean interpolates from the front-end estimate to the target channel, while the covariance shrinks to zero as $t\to1$.

For the path in \eqref{eq:refiner_interpolation}, the target transport vector is
\begin{equation*}
\mathbf{w} = \mathbf{h} - \bigl(\tilde{\mathbf{h}}+\sigma_0\boldsymbol{\epsilon}\bigr),
\end{equation*}
which is constant for a fixed tuple $(\mathbf{h},\mathbf{b},\boldsymbol{\epsilon})$. The refiner is trained using the conditional flow-matching objective
\begin{equation}
\mathcal{L}_{\mathrm{ref}} =
\mathbb{E}_{\substack{ (\mathbf{h},\mathbf{b}),\\
t\sim\mathcal{U}(0,1),\\
\boldsymbol{\epsilon}\sim\mathcal{N}(\mathbf{0},\mathbf{I})}}
\!\left[ \left\| \boldsymbol{\phi}_{\theta}(\mathbf{h}(t),\mathbf{b},t)
- \mathbf{w} \right\|_2^2 \right].
\label{eq:refiner_loss}
\end{equation}

At inference, the refined estimate is obtained by solving \eqref{eq:refiner_ode} over $t\in[0,1]$ with deterministic initial condition $\mathbf{h}(0)=\tilde{\mathbf{h}}$ with $\mathbf{\epsilon}=\mathbf{0}$. 
In our implementation, we use the second-order midpoint method to solve the ODE numerically \cite{Welker25}. 
Let $N_{\mathrm{step}}$ denote the number of flow steps, with $\Delta t=1/N_{\mathrm{step}}$ and $t_n=n\Delta t$. 
We use the midpoint method to solve the learned ODE. 
At step $n$, we first evaluate the vector field at the current state as $\mathbf{k}_1=\boldsymbol{\phi}_{\theta}(\mathbf{h}_n,\mathbf{b},t_n)$, and then evaluate it at the midpoint state and midpoint flow time as $\mathbf{k}_2=\boldsymbol{\phi}_{\theta}(\mathbf{h}_n+\frac{\Delta t}{2}\mathbf{k}_1,\mathbf{b},t_n+\frac{\Delta t}{2})$. 
The update is then
\begin{equation*}
\mathbf{h}_{n+1} = \mathbf{h}_n+\Delta t\,\mathbf{k}_2.
\end{equation*}
The terminal state $\hat{\mathbf{h}}=\mathbf{h}(1)$ is taken as the final refined CSI estimate.

\subsection{Direct Conditional Flow Decoder}
\label{sec:direct_flow_decoder}

We next consider a fully generative BS-side decoder, referred to as the \emph{direct conditional flow decoder}. Unlike the flow-based CSI refiner, which first forms a deterministic front-end estimate $\tilde{\mathbf{h}}=D_0(\mathbf{b})$ and then refines it, the direct decoder reconstructs CSI in a single stage by conditioning the flow directly on the received feedback bits $\mathbf{b}$ and generating the channel from a Gaussian prior, as in diffusion-type generative models \cite{Ho20}.

The direct decoder uses the same feedback-conditioned target distribution as the refiner, $q_1(\mathbf{h}\mid \mathbf{b}) = p(\mathbf{h}\mid \mathbf{b})$, but replaces the front-end-centered source with a standard Gaussian prior, $q_0(\mathbf{h}\mid \mathbf{b}) = \mathcal{N}(\mathbf{0},\mathbf{I})$.

The training path and flow-matching objective follow the same form as in Section~\ref{sec:flowdec_refiner}, with the front-end-centered source sample $\tilde{\mathbf{h}}+\sigma_0\boldsymbol{\epsilon}$ replaced by $\boldsymbol{\epsilon}\sim\mathcal{N}(\mathbf{0},\mathbf{I})$.
Equivalently, for each training pair $(\mathbf{h},\mathbf{b})$ and $t\sim\mathcal{U}(0,1)$, the path becomes $\mathbf{h}(t)=(1-t)\boldsymbol{\epsilon}+t\mathbf{h}$, yielding the target transport vector $\mathbf{w}=\mathbf{h}-\boldsymbol{\epsilon}$.

At inference, given $\mathbf{b}$, the decoder differs from the refiner only in its initialization: instead of starting from deterministic $\tilde{\mathbf{h}}$, it starts from $\mathbf{h}(0)=\boldsymbol{\epsilon}$ with $\boldsymbol{\epsilon}\sim\mathcal{N}(\mathbf{0},\mathbf{I})$.
The subsequent flow ODE is solved in the same manner as in the refiner, and the terminal state $\hat{\mathbf{h}}=\mathbf{h}(1)$ is taken as the reconstructed CSI.

\subsection{Vector-Field Network and Implementation Details}
\label{sec:vector_field_network}

\begin{figure}[t]
\centering
\subfigure[Residual block used in the vector-field network.]{
    \includegraphics
    [width=2.9in]{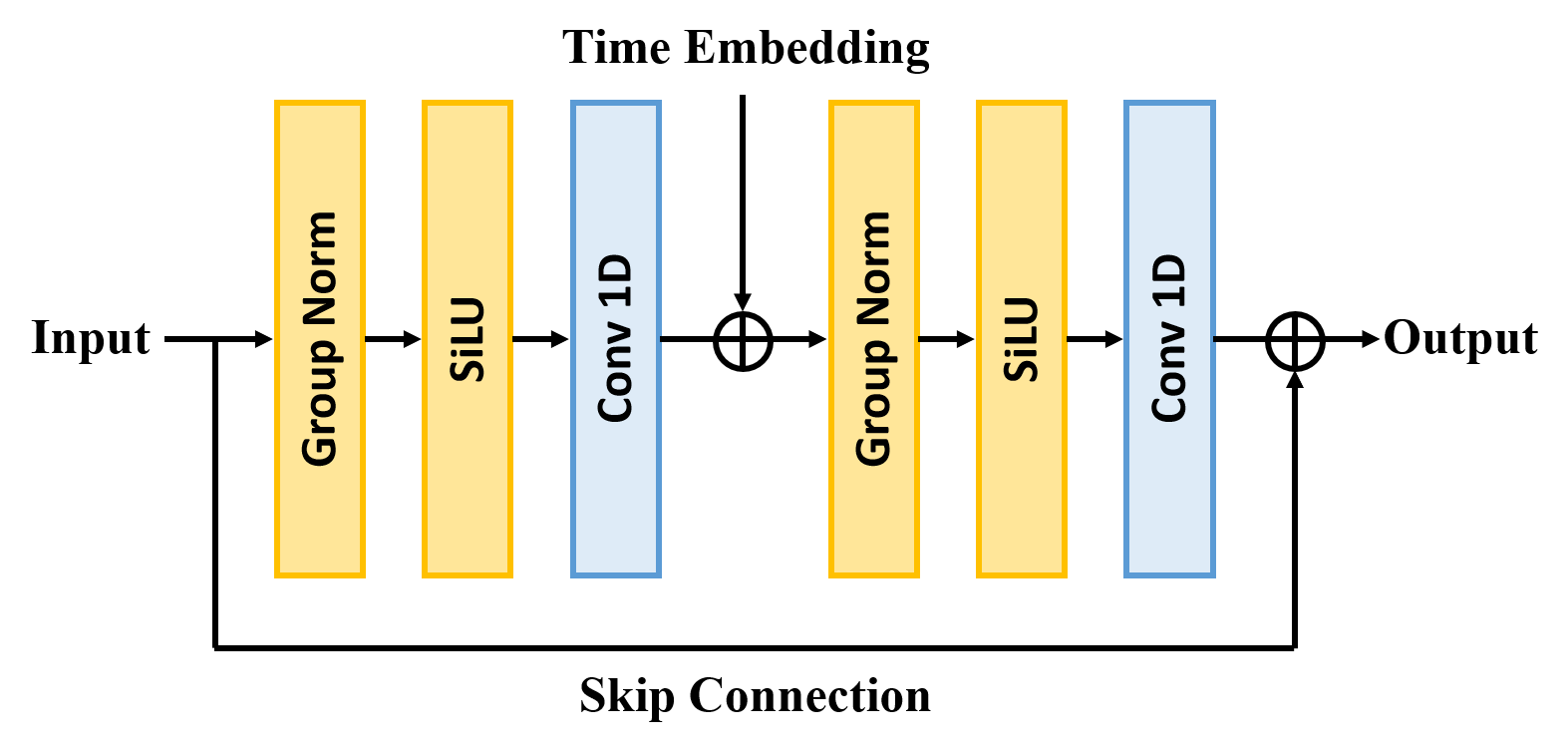}
    \label{fig:resblock}}
\vspace{0.5em}
\subfigure[Overall architecture of the proposed 1D U-Net vector-field network with feedback bits conditioning and time-embedding injection.]{
    \includegraphics
    [width=3.4in]{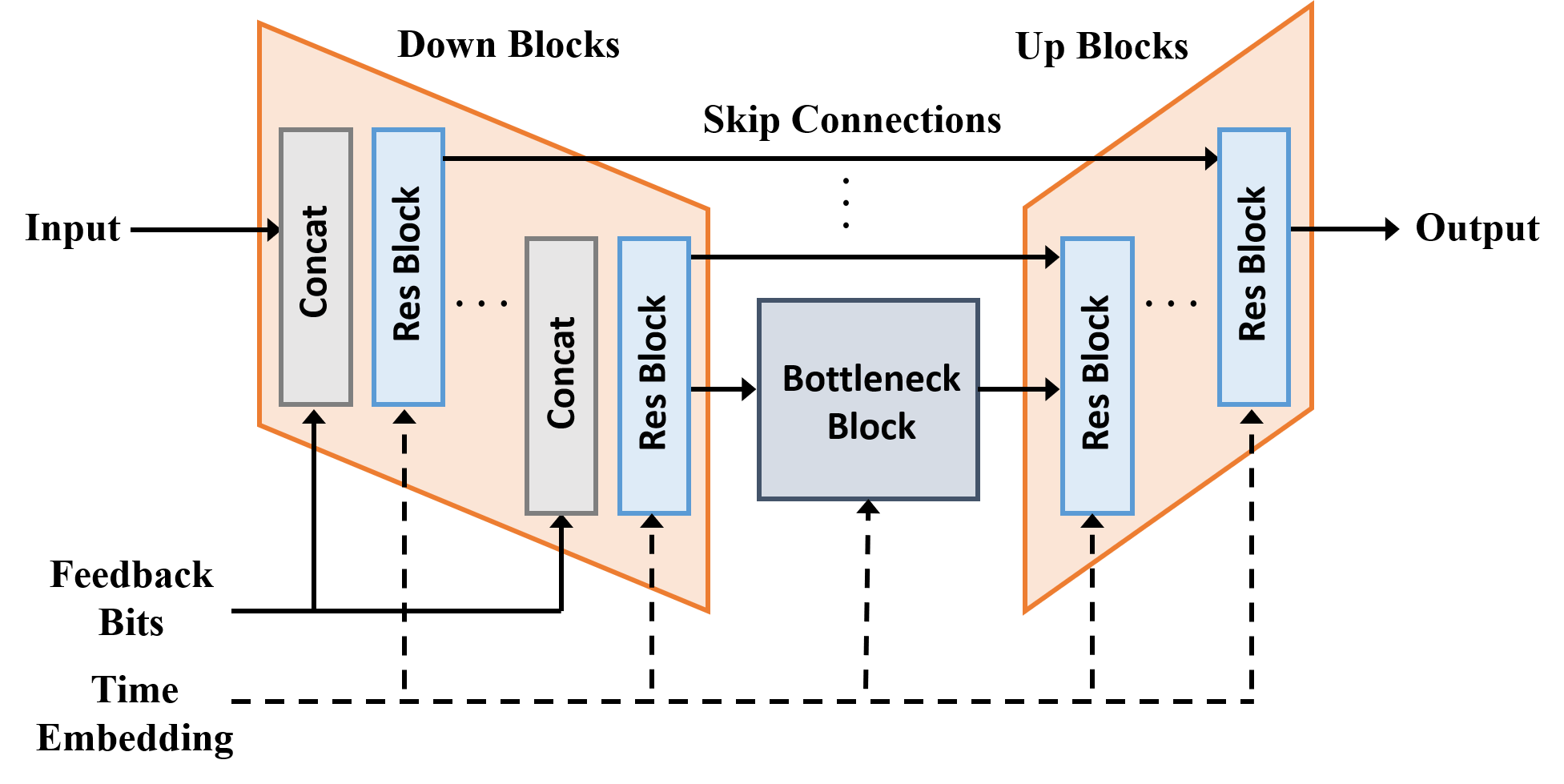}
    \label{fig:unet}}
\caption{Network architecture of the proposed vector-field model.}
\label{fig:network_architecture}
\end{figure}

Both the flow-based CSI refiner and the direct conditional flow decoder use the same neural vector-field backbone. 
Specifically, $\boldsymbol{\phi}_{\theta}$ is implemented by a multi-scale one-dimensional (1D) U-Net, a U-shaped convolutional architecture with skip connections between matching resolutions.
The complex CSI vector is represented as a $2\times N$ real-valued tensor, whose two rows correspond to real and imaginary feature maps, while the 1D convolutions operate along the antenna dimension. 
This architecture can capture local spatial correlations at multiple scales.
Let $N_{\mathrm{down}}$ and $N_{\mathrm{up}}$ denote the numbers of residual blocks in each downsampling and upsampling block, respectively. 
As illustrated in Fig.~\ref{fig:resblock}, each residual block consists of GroupNorm, SiLU, and two 1D convolutional layers, with the time embedding added after the first convolution.

The conditioning information derived from the feedback bits $\mathbf{b}$ is incorporated by concatenating it with the current intermediate state at the network input. This conditioning representation is then progressively downsampled and reintroduced into the downsampling path at multiple resolutions \cite{Welker25}, allowing the network to preserve feedback-dependent information across scales rather than injecting it only once at the input layer.

The vector field depends on the flow time $t$, so the network is conditioned on $t$ at each ODE step through a time embedding. We encode $t$ using Gaussian Fourier features \cite{Song21}. Specifically, let $\mathbf{W}\sim\mathcal{N}(\mathbf{0},\sigma_f^2\mathbf{I})$ be fixed after initialization. For $t\in[0,1]$, the Fourier feature is defined as
\begin{equation}
\gamma(t) = \left[ \sin(2\pi t\mathbf{W}),\; \cos(2\pi t\mathbf{W}) \right].
\end{equation}
This feature is processed by a two-layer multilayer perceptron to produce a time embedding, which is linearly projected and added to each residual block after the first convolution. 
Group normalization and SiLU activations are used throughout the network.

To reduce the effect of short-term fluctuations in the training weights, we maintain an exponential moving average (EMA) of the network parameters, which can be expressed as
\begin{equation}
\bar{\theta} \leftarrow
\beta\,\bar{\theta} + (1-\beta)\theta,
\end{equation}
where $\beta\in(0,1)$ is a decay factor close to one. Following common practice in diffusion-based generative modeling, the EMA parameters are used for validation and inference, since they typically yield more stable and reliable predictions than the instantaneous training weights \cite{Ho20}.

\begin{algorithm}[t]
\caption{Inference Procedure of the Proposed Flow-Based CSI Reconstruction}
\label{alg:flow_inference}
\begin{algorithmic}
\State \textbf{Initialize:} feedback bits $\mathbf{b}$, trained vector-field $\boldsymbol{\phi}_{\bar{\theta}}$, number of flow steps $N_{\mathrm{step}}$, mode $\in\{\text{refiner},\text{direct decoder}\}$
\If{mode = refiner}
    \State Set $\mathbf{h}(0)=\tilde{\mathbf{h}}$ with $\tilde{\mathbf{h}}=D_0(\mathbf{b})$ 
\Else
    \State Set $\mathbf{h}(0)=\boldsymbol{\epsilon}$ by sampling $\boldsymbol{\epsilon}\sim\mathcal{N}(\mathbf{0},\mathbf{I})$
\EndIf
\State Set $\Delta t=1/N_{\mathrm{step}}$
\For{$n=0,\dots,N_{\mathrm{step}}-1$}
    \State Set $t_n=n\Delta t$
    \State Compute $\mathbf{k}_1=\boldsymbol{\phi}_{\bar{\theta}}(\mathbf{h}_n,\mathbf{b},t_n)$
    \State Compute $\mathbf{k}_2=\boldsymbol{\phi}_{\bar{\theta}}\!\left(\mathbf{h}_n+\frac{\Delta t}{2}\mathbf{k}_1,\mathbf{b},t_n+\frac{\Delta t}{2}\right)$
    \State Update $\mathbf{h}_{n+1}=\mathbf{h}_n+\Delta t\,\mathbf{k}_2$
\EndFor
\State \textbf{Return:} the reconstructed CSI $\hat{\mathbf{h}}=\mathbf{h}_{N_{\mathrm{step}}}$
\end{algorithmic}
\end{algorithm}

\section{Posterior Directional Geometry under Limited CSI Feedback}
\label{sec:posterior_directional_geometry}

This section studies CSI reconstruction under finite-rate feedback from a posterior directional-geometry viewpoint. 
By decomposing the SINR into desired-signal and interference terms, we show that posterior geometry affects these two components in different ways. 
The analysis proceeds in two steps. 
First, we study the desired-signal term and characterize the posterior-optimal direction for the feedback-conditioned self-alignment criterion, clarifying when the MSE-oriented conditional mean becomes suboptimal.
Second, we turn to the interference term and use a finite-mixture posterior surrogate to compare posterior sampling with conditional-mean averaging under ZF precoding.

\subsection{Feedback-Conditioned Self-Alignment and Interference}
\label{subsec:chordal_and_interference}

The feedback mapping $E:\mathbb{C}^{N}\to\{0,1\}^B$ is generally many-to-one. 
For a feedback vector $\mathbf{b}$, we denote the corresponding feedback cell as
\begin{equation}
\mathcal{C}(\mathbf{b})\triangleq\{\mathbf{h}:E(\mathbf{h})=\mathbf{b}\}.  
\label{eqn:def_fb_cell}
\end{equation}
Thus, a received feedback vector specifies a posterior distribution over plausible channels in this cell rather than a unique channel realization. 
The key issue is therefore not perfect channel recovery, but whether the reconstruction error is shaped favorably for the resulting sum-rate.

Under ZF precoding, the relevant geometric object is the unit-norm channel direction. Decomposing $\mathbf{h}_k=\rho_k\mathbf{u}_k$, where $\rho_k=\|\mathbf{h}_k\|_2$ and $\|\mathbf{u}_k\|_2=1$, \eqref{eqn:sinr} can be rewritten as
\begin{equation*}
R_k = \log_2 \left(1+\frac{\rho_k^2\left|\mathbf{u}_k^H\mathbf{f}_k\right|^2}{\sigma^2+\rho_k^2\sum_{n\neq k}\left|\mathbf{u}_k^H\mathbf{f}_n\right|^2}\right).
\end{equation*}
For fixed $\rho_k$, the rate is improved by increasing the desired-signal gain $\left|\mathbf{u}_k^H\mathbf{f}_k\right|^2$ while suppressing the aggregate interference $\sum_{n\neq k}\left|\mathbf{u}_k^H\mathbf{f}_n\right|^2$.

We first consider the desired-signal term. 
For user $k$, the directional uncertainty given the received feedback bits $\mathbf{b}_k$ is described by the conditional second-moment matrix 
$\mathbf{R}(\mathbf{b}_k)\triangleq\mathbb{E}[\mathbf{u}_k\mathbf{u}_k^H\mid\mathbf{b}_k]$.
For a fixed beam $\mathbf{f}_k$, the posterior expected directional gain of user $k$ is
\begin{equation}
\mathbb{E}\left[\left|\mathbf{u}_k^H \mathbf{f}_k\right|^2 \mid \mathbf{b}_k\right] =
\mathbf{f}_k^H \mathbf{R}(\mathbf{b}_k) \mathbf{f}_k.
\label{eq:desired_signal}
\end{equation}
This expression is coupled through the ZF beam $\mathbf{f}_k$, which is jointly determined by all reconstructed user directions $\{\hat{\mathbf{u}}_j\}_{j=1}^K$. Therefore, it does not directly reveal which representative direction is favorable within the feedback cell of user $k$. We separate this single-user directional component from the subsequent ZF coupling using the feedback-conditioned self-alignment criterion
\begin{equation}
\mathbb{E}\!\left[\left|\mathbf{u}_k^H\hat{\mathbf{u}}_k\right|^2 \mid \mathbf{b}_k\right] = \hat{\mathbf{u}}_k^H\mathbf{R}(\mathbf{b}_k)\hat{\mathbf{u}}_k .
\label{eq:local_self_alignment_criterion}
\end{equation}
This criterion leads directly to the conditional expected chordal distortion
\begin{equation}
\mathbb{E}\!\left[d_c^2\bigl(\mathbf{u}_k,\hat{\mathbf{u}}_k\bigr) \,\middle|\, \mathbf{b}_k \right] = 1-\hat{\mathbf{u}}_k^H\mathbf{R}(\mathbf{b}_k)\hat{\mathbf{u}}_k,
\label{eq:chordal_R_relation}
\end{equation}
where
$d_c^2\bigl(\mathbf{u}_k,\hat{\mathbf{u}}_k\bigr) \triangleq 1-\left|\mathbf{u}_k^H\hat{\mathbf{u}}_k\right|^2$.
Thus, maximizing the self-alignment criterion is equivalent to minimizing the conditional expected chordal distortion, matching the chordal-distance measure commonly used for single-user beamforming with limited feedback. 
Although this criterion is closely related to the desired-signal gain when ZF beam synthesis only mildly perturbs the intended direction, it remains a local per-user criterion in MU-MIMO.

Next, we turn to the interference term. 
Compared to the desired-signal term, which involves the beam intended for user $k$, the interference term involves all beams intended for the other users. 
For each $n\neq k$, the ZF beam $\mathbf{f}_n$ satisfies $\hat{\mathbf{u}}_k^H\mathbf{f}_n=0$, but the actual leakage $\left|\mathbf{u}_k^H\mathbf{f}_n\right|^2$ remains whenever $\mathbf{u}_k$ differs from $\hat{\mathbf{u}}_k$. 
Therefore, interference depends on the interaction between the reconstruction error of user $k$ and the collection of other-user ZF beams, rather than on a per-user chordal criterion alone.

With $\mathbf{B}\triangleq[\mathbf{b}_1,\ldots,\mathbf{b}_K]$ denoting the feedback bits from the scheduled UEs, the resulting total interference metric is
\begin{align}
\mathcal{I}(\mathbf{B};\hat{\mathbf{u}}_1,\ldots,\hat{\mathbf{u}}_K)
& \triangleq \sum_{k=1}^K \sum_{n\neq k} \mathbb{E}\!\left[ \left| \mathbf{u}_k^H \mathbf{f}_n \right|^2 \,\middle|\, \mathbf{B}\right] \nonumber\\
& = \sum_{k=1}^K \sum_{n\neq k} \mathbf{f}_n^H\mathbf{R}(\mathbf{b}_k)\mathbf{f}_n .
\label{eq:zf_interference_general}
\end{align}
In summary, desired-signal geometry admits a self-alignment criterion, whereas interference geometry remains coupled through the ZF beams and depends on the orientation of directional errors relative to the nulling directions.

\subsection{Posterior-Optimal Direction for Self-Alignment}
\label{subsec:bayes_optimal_direction}

Section~\ref{subsec:chordal_and_interference} introduced the self-alignment criterion in \eqref{eq:local_self_alignment_criterion}, which captures the per-user directional component of the desired-signal geometry. We now characterize the posterior-optimal representative direction for self-alignment. For an arbitrary feedback vector $\mathbf{b}$, let
$\mathbf{R}(\mathbf{b})\triangleq\mathbb{E}[\mathbf{u}\mathbf{u}^H\mid\mathbf{b}]$.

\begin{theorem}[Posterior-optimal direction for self-alignment]
\label{thm:chordal_bayes_opt}
Let $\lambda_{\max}$ denote the largest eigenvalue of $\mathbf{R}(\mathbf{b})$. The direction minimizing the conditional expected chordal distortion is any principal eigenvector of $\mathbf{R}(\mathbf{b})$, that is,
\begin{equation*}
\arg\min_{\|\hat{\mathbf{u}}\|_2=1} \mathbb{E}\!\left[ d_c^2\bigl(\mathbf{u},\hat{\mathbf{u}}\bigr) \,\middle|\,
\mathbf{b} \right] =
\arg\max_{\|\hat{\mathbf{u}}\|_2=1} \hat{\mathbf{u}}^H \mathbf{R}(\mathbf{b}) \hat{\mathbf{u}}.
\end{equation*}
Moreover,
\begin{equation}
\max_{\|\hat{\mathbf{u}}\|_2=1}
\hat{\mathbf{u}}^H \mathbf{R}(\mathbf{b}) \hat{\mathbf{u}}
= \lambda_{\max}
\label{eq:self_alignment_max}
\end{equation}
and
\begin{equation}
\min_{\|\hat{\mathbf{u}}\|_2=1}
\mathbb{E}\!\left[ d_c^2\bigl(\mathbf{u},\hat{\mathbf{u}}\bigr)
\,\middle|\, \mathbf{b}\right] = 1-\lambda_{\max}.
\label{eq:min_chordal_bayes}
\end{equation}
\end{theorem}

\begin{proof}
From \eqref{eq:chordal_R_relation},
\begin{equation*}
\mathbb{E}\!\left[
d_c^2\bigl(\mathbf{u},\hat{\mathbf{u}}\bigr) \,\middle|\, \mathbf{b} \right] =
1-\hat{\mathbf{u}}^H \mathbf{R}(\mathbf{b}) \hat{\mathbf{u}}.
\end{equation*}
Thus, minimizing the conditional expected chordal distortion over $\|\hat{\mathbf{u}}\|_2=1$ is equivalent to maximizing the Rayleigh quotient
$\hat{\mathbf{u}}^H \mathbf{R}(\mathbf{b}) \hat{\mathbf{u}}$.
The maximum is $\lambda_{\max}$ and is attained by any principal eigenvector of $\mathbf{R}(\mathbf{b})$, which proves \eqref{eq:self_alignment_max}. Substituting this maximum into \eqref{eq:chordal_R_relation} yields \eqref{eq:min_chordal_bayes}.
\end{proof}

Theorem~\ref{thm:chordal_bayes_opt} identifies the posterior-optimal direction for the feedback-conditioned self-alignment criterion. 
The optimal direction is given by the dominant eigenvector of the posterior second-moment matrix $\mathbf{R}(\mathbf{b})$. 
While this criterion is optimal for single-user beamforming alignment, MU-MIMO precoding remains coupled across users through the ZF beams.
We next compare this posterior-optimal direction with the direction obtained from the MSE-optimal estimate.

\begin{figure}[t]
    \centering
    \includegraphics
    [width=3.3in]
    {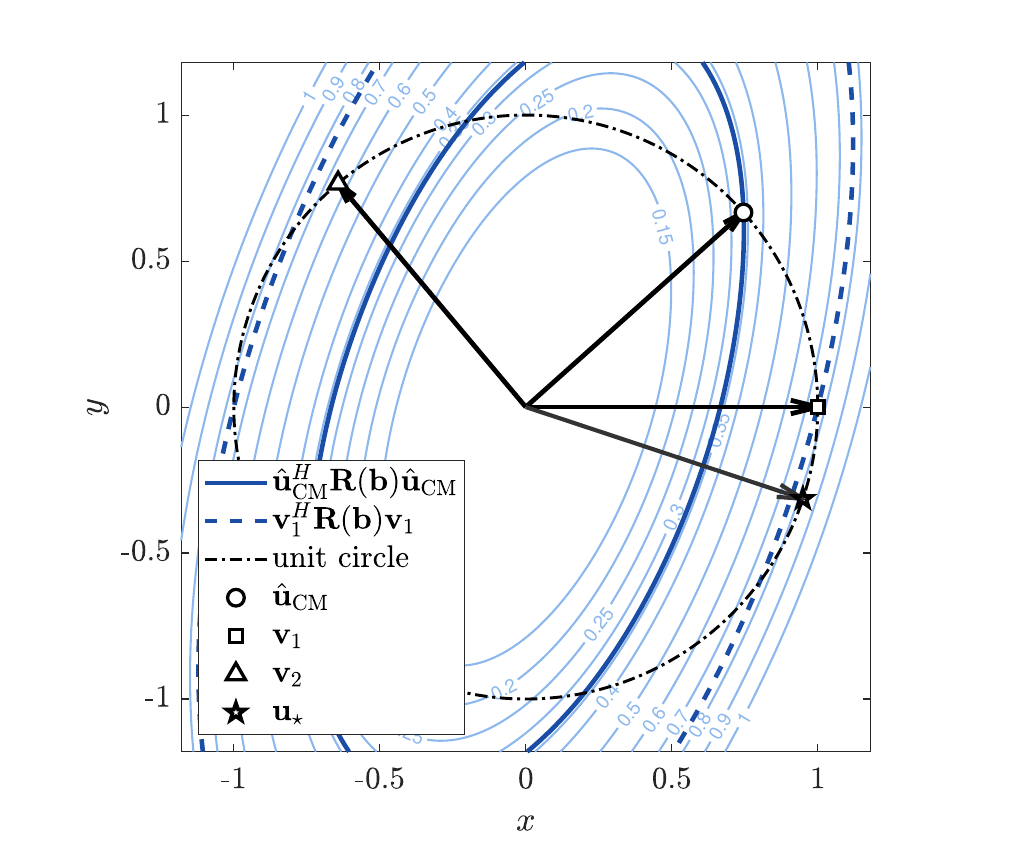}
    \caption{Toy example under the posterior model $0.6\,\mathcal{N}(\mathbf{v}_1,0.01\mathbf{I}) + 0.4\,\mathcal{N}(\mathbf{v}_2,0.01\mathbf{I})$ with $\mathbf{v}_1=[1,0]^T$ and $\mathbf{v}_2=[\cos 130^\circ,\sin 130^\circ]^T$, where the contours of $\mathbf{u}^T \mathbf{R}\mathbf{u}$ show that $\hat{\mathbf{u}}_{\mathrm{CM}}$ is worse than $\mathbf{v}_1$ $(0.337$ versus $0.760)$, while the optimal direction ${\mathbf{u}}_{\star}$ attains $0.824$.}
    \label{fig:toy_q_contour_cm_suboptimal}
\end{figure}

\subsection{MSE Averaging Versus Posterior-Guided Reconstruction}
\label{subsec:conditional_mean_suboptimality}

The optimal direction identified in Theorem~\ref{thm:chordal_bayes_opt} is governed by the posterior second-moment matrix $\mathbf{R}(\mathbf{b})$. In practice, however, many CSI decoders are trained to minimize the Euclidean MSE of the channel reconstruction. Under this criterion, the Bayes-optimal estimator is the conditional mean estimator \cite{Kay93}, given by
\begin{equation*}
\hat{\mathbf{h}}_{\mathrm{CM}}(\mathbf{b}) = \mathbb{E}\left[\mathbf{h}\mid \mathbf{b}\right].
\end{equation*}
For the directional criteria above, we evaluate this estimator through its normalized direction
$\hat{\mathbf{u}}_{\mathrm{CM}}(\mathbf{b})
\triangleq
{\hat{\mathbf{h}}_{\mathrm{CM}}(\mathbf{b})}/
{\|\hat{\mathbf{h}}_{\mathrm{CM}}(\mathbf{b})\|_2}$,
assuming $\hat{\mathbf{h}}_{\mathrm{CM}}(\mathbf{b})\neq\mathbf{0}$.

The conditional-mean direction is particularly vulnerable when the posterior $p(\mathbf{u}\mid\mathbf{b})$ is multimodal and places comparable mass on several well-separated plausible directions. In this case, Euclidean averaging may produce a centroid-like compromise that does not preserve the underlying channel geometry. As illustrated in Fig.~\ref{fig:toy_q_contour_cm_suboptimal}, this can make $\hat{\mathbf{u}}_{\mathrm{CM}}$ inferior to mode-aligned directions under the chordal-gain criterion in \eqref{eq:chordal_R_relation}. Section~\ref{subsec:zf_cell_to_cell} further examines this averaging effect through the total interference metric $\mathcal{I}(\mathbf{B};\hat{\mathbf{u}}_1,\ldots,\hat{\mathbf{u}}_K)$.

This posterior-geometry viewpoint also clarifies why the proposed flow-based decoders can better preserve channel geometry. They learn feedback-conditioned transport toward the channel posterior associated with $\mathbf{b}$, encouraging reconstructions to follow posterior channel structure rather than averaging across distinct plausible directions.

For the refiner in Section~\ref{sec:flowdec_refiner}, the transport is anchored at the front-end reconstruction $\tilde{\mathbf{h}}(\mathbf{b})=D_0(\mathbf{b})$. During training, the source sample is perturbed as
\begin{equation*}
\hat{\mathbf{h}}_{\Omega}(\mathbf{b};\boldsymbol{\epsilon}) \triangleq
\Omega_\theta\bigl( \tilde{\mathbf{h}}(\mathbf{b})+\sigma_0\boldsymbol{\epsilon}, \mathbf{b} \bigr).
\end{equation*}
At inference, we use $\boldsymbol{\epsilon}=\mathbf{0}$, so the refiner is deterministic. Nevertheless, it remains posterior-guided because its learned transport targets $p(\mathbf{h}\mid\mathbf{b})$ rather than the Euclidean conditional mean.

In contrast, the direct flow decoder in Section~\ref{sec:direct_flow_decoder},
\begin{equation*}
\hat{\mathbf{h}}_{D}(\mathbf{b};\boldsymbol{\epsilon}) \triangleq
D_\theta(\mathbf{b};\boldsymbol{\epsilon}),
\qquad
\boldsymbol{\epsilon}\sim\mathcal{N}(\mathbf{0},\mathbf{I}),
\end{equation*}
uses random initialization to generate posterior-consistent channel realizations. 
Thus, the direct decoder can represent posterior variability through different initial noise samples, helping avoid collapse into a single averaged reconstruction. 
This stochasticity, however, does not make an arbitrary sample optimal for sum-rate. 
The refiner instead uses perturbation as a training regularizer and sets $\boldsymbol{\epsilon}=\mathbf{0}$ at inference for stable deterministic refinement.

\subsection{ZF Interference Geometry: Geometry Collapse, Residual Overlap, and Separation of the Directions}
\label{subsec:zf_cell_to_cell}

The preceding analysis identifies which direction is locally favorable for the desired-signal term. 
We now turn to the interference term, where the key issue is not only whether each user is well aligned with its own posterior, but also how the reconstructed directions shape the nulls used for inter-user interference suppression. 
To study this effect, we compare stochastic posterior sampling with deterministic conditional-mean averaging by approximating each conditional posterior $p(\mathbf{u}_k\mid\mathbf{b}_k)$ by a finite mixture over representative modes.

Specifically, conditioned on $\mathbf{b}_k$, let $\{\mathbf{v}_k(m)\}_{m=1}^M$ denote the representative directions, with $\|\mathbf{v}_k(m)\|_2=1$, and let $\{p_k(m)\}_{m=1}^M$ denote the corresponding posterior weights, satisfying $\sum_{m=1}^M p_k(m)=1$. For notational simplicity, we let all users share the same number $M$ of representative modes.

Under posterior sampling, user $k$ selects a mode index $m_k\in\{1,\ldots,M\}$ according to the conditional law $p(m_k\mid\mathbf b_k)=p_k(m_k)$.
The selected mode indices are independent across users. Thus, for $\mathbf{m}\!=\!(m_1,\ldots,m_K)\!\in\!\{1,\ldots,M\}^K$, we write
\begin{equation*}
p(\mathbf{m}\mid\mathbf{B}) = \prod_{k=1}^K p_k(m_k).
\end{equation*}
For a selected mode tuple $\mathbf{m}$, define
\begin{equation*}
\mathcal{U}_n(\mathbf{m})
\triangleq
\operatorname{span}\{\mathbf{v}_j(m_j):j\neq n\},
\end{equation*}
and let $\mathbf{\Pi}_n(\mathbf{m})$ denote the orthogonal projector onto $\mathcal{U}_n(\mathbf{m})^\perp$. Then, the normalized ZF beam for user $n$ is
$\mathbf{f}_n(\mathbf{m}) =
{\mathbf{\Pi}_n(\mathbf{m})\mathbf{v}_n(m_n)}/
{\|\mathbf{\Pi}_n(\mathbf{m})\mathbf{v}_n(m_n)\|_2}$.
For a competing user $k\neq n$, only the non-selected modes contribute, since $\mathbf{v}_k(m_k)\in\mathcal{U}_n(\mathbf{m})$ implies $\mathbf{\Pi}_n(\mathbf{m})\mathbf{v}_k(m_k)=\mathbf{0}$.
Thus, for fixed $n$ and $k$, the corresponding posterior-averaged interference is
$\sum_{\ell\neq m_k}
p_k(\ell)
\left|
\mathbf{v}_k(\ell)^H\mathbf{f}_n(\mathbf{m})
\right|^2$.
Substituting $\mathbf{f}_n(\mathbf{m})$ into this interference term and summing over all beams $n$ and competing users $k\neq n$ yields the total interference under the selected tuple
\begin{equation*}
J_{\mathbf{m}}
\triangleq
\sum_{n=1}^K \sum_{k\neq n} \frac{
\sum_{\ell\neq m_k} p_k(\ell) \left| \mathbf{v}_k(\ell)^H
\mathbf{\Pi}_n(\mathbf{m}) \mathbf{v}_n(m_n) \right|^2}{\|\mathbf{\Pi}_n(\mathbf{m})\mathbf{v}_n(m_n)\|_2^2}.
\end{equation*}
Thus, once one mode per user has been selected, only the non-selected posterior mass contributes to the remaining interference, and averaging over $\mathbf{m}$ gives
\begin{equation*}
\mathcal{I}_{\mathrm{PS}}(\mathbf{B}) = \sum_{\mathbf{m}} p(\mathbf{m}\mid\mathbf{B})\,J_{\mathbf{m}}.
\end{equation*}
This expression shows that posterior sampling nulls the sampled mode of each user, but residual leakage can remain from other plausible modes in the same feedback cell.

For the conditional-mean benchmark, we form the conditional mean
$\boldsymbol{\mu}_k \triangleq \sum_{m=1}^M p_k(m)\mathbf{v}_k(m)$
and the associated unit-norm direction
$\hat{\mathbf{u}}_{k,\mathrm{CM}} \triangleq \boldsymbol{\mu}_k/\|\boldsymbol{\mu}_k\|_2$,
assuming $\boldsymbol{\mu}_k\neq\mathbf{0}$. Analogously to the posterior-sampling case, for user $n$, let
\begin{equation*}
\mathcal{U}_{n,\mathrm{CM}} \triangleq
\operatorname{span}\{\hat{\mathbf{u}}_{j,\mathrm{CM}}:j\neq n\},
\end{equation*}
and let $\mathbf{\Pi}_{n,\mathrm{CM}}$ denote the orthogonal projector onto $(\mathcal{U}_{n,\mathrm{CM}})^\perp$. The resulting interference under conditional-mean averaging is
\begin{equation}
\label{eq:Icm_general_k}
\mathcal{I}_{\mathrm{CM}}(\mathbf{B}) =
\sum_{n=1}^K \sum_{k\neq n} \frac{ \sum_{m=1}^M p_k(m)
\left| \mathbf{v}_k(m)^H \mathbf{\Pi}_{n,\mathrm{CM}} \hat{\mathbf{u}}_{n,\mathrm{CM}} \right|^2}{\|\mathbf{\Pi}_{n,\mathrm{CM}}\hat{\mathbf{u}}_{n,\mathrm{CM}}\|_2^2}.
\end{equation}
In contrast, conditional-mean averaging nulls the centroid-like direction of each feedback cell, but residual leakage can remain from dominant posterior modes that are not aligned with this averaged direction.

The comparison between $\mathcal{I}_{\mathrm{PS}}(\mathbf{B})$ and $\mathcal{I}_{\mathrm{CM}}(\mathbf{B})$ is governed by three effects: geometry collapse under conditional-mean averaging, residual overlap under posterior sampling, and the separation of the resulting directions. We next expose these three effects through a local comparison of the pairwise terms.

For a given tuple $\mathbf{m}$ and a pair $(k,n)$ with $k\neq n$, the posterior-sampling local term is
\begin{equation}
\mathcal{I}_{k,n}^{\mathrm{PS}}(\mathbf{m})
\triangleq
\frac{\mathcal{R}_{k,n}^{\mathrm{PS}}(\mathbf{m})}{
\eta_{n,\mathrm{PS}}(\mathbf{m})},
\label{eq:local_ps_exact}
\end{equation}
where
$\mathcal{R}_{k,n}^{\mathrm{PS}}(\mathbf{m})
\triangleq
\sum_{\ell\neq m_k} p_k(\ell) \left| \mathbf{v}_k(\ell)^H
\mathbf{\Pi}_n(\mathbf{m})\mathbf{v}_n(m_n) \right|^2$
and
$\eta_{n,\mathrm{PS}}(\mathbf{m}) \triangleq
\left\| \mathbf{\Pi}_n(\mathbf{m})\mathbf{v}_n(m_n) \right\|_2^2$. Thus, $J_{\mathbf{m}}=\sum_{n=1}^{K}\sum_{k\neq n}
\mathcal{I}_{k,n}^{\mathrm{PS}}(\mathbf{m})$, and
$\mathcal{I}_{\mathrm{PS}}(\mathbf{B})$ is obtained by averaging $J_{\mathbf{m}}$ over
$\mathbf{m}$ according to $p(\mathbf{m}\mid\mathbf{B})$.

To compare conditional-mean averaging with posterior sampling for a selected tuple $\mathbf{m}$, we rewrite the fixed conditional-mean direction into the selected-mode contribution and the remaining averaged contribution as
\begin{equation}
\hat{\mathbf{u}}_{n,\mathrm{CM}} = 
c_n(\mathbf{m})\,\mathbf{v}_n(m_n) + \mathbf{d}_n(\mathbf{m}),
\label{eq:cm_branch_decomposition}
\end{equation}
where
$c_n(\mathbf{m}) \triangleq \frac{p_n(m_n)}{\|\boldsymbol{\mu}_n\|_2}$,
$\mathbf{d}_n(\mathbf{m}) \triangleq \sum_{r\neq m_n}
\frac{p_n(r)}{\|\boldsymbol{\mu}_n\|_2}\mathbf{v}_n(r)$,
and
$\boldsymbol{\mu}_n=\sum_{r=1}^M p_n(r)\mathbf{v}_n(r)$.
Then, the corresponding conditional-mean local term is
\begin{align}
\mathcal{I}_{k,n}^{\mathrm{CM}}
&=
\frac{1}{\eta_{n,\mathrm{CM}}} \Big( \mathcal{S}_{k,n}^{\mathrm{CM}}(\mathbf{m})
\nonumber\\
&+
\sum_{\ell\neq m_k} p_k(\ell) \left| c_n(\mathbf{m})\,a_{k,n,\ell}^{\mathrm{CM}}(\mathbf{m}) + b_{k,n,\ell}^{\mathrm{CM}}(\mathbf{m}) \right|^2\Big),
\label{eq:local_cm_exact}
\end{align}
where
$a_{k,n,\ell}^{\mathrm{CM}}(\mathbf{m}) \triangleq \mathbf{v}_k(\ell)^H\mathbf{\Pi}_{n,\mathrm{CM}}\mathbf{v}_n(m_n)$
and
$b_{k,n,\ell}^{\mathrm{CM}}(\mathbf{m})
\triangleq
\mathbf{v}_k(\ell)^H\mathbf{\Pi}_{n,\mathrm{CM}}\mathbf{d}_n(\mathbf{m})$.
Here,
\begin{equation*}
\mathcal{S}_{k,n}^{\mathrm{CM}}(\mathbf{m}) \triangleq p_k(m_k)
\left| \mathbf{v}_k(m_k)^H\mathbf{\Pi}_{n,\mathrm{CM}}\hat{\mathbf{u}}_{n,\mathrm{CM}} \right|^2,
\end{equation*}
which quantifies the leakage of the selected competitor branch under the conditional-mean projector itself, and
\begin{equation*}
\eta_{n,\mathrm{CM}} \triangleq \left\| \mathbf{\Pi}_{n,\mathrm{CM}}\hat{\mathbf{u}}_{n,\mathrm{CM}} \right\|_2^2,
\end{equation*}
which quantifies the projected energy of the resulting conditional-mean direction.

Thus, \eqref{eq:local_ps_exact} and \eqref{eq:local_cm_exact} expose the three effects directly: $\mathcal{R}_{k,n}^{\mathrm{PS}}(\mathbf{m})$ captures the residual overlap left under posterior sampling, $\mathbf{d}_n(\mathbf{m})$ captures the geometry-collapsing effect of conditional-mean averaging, and $\eta_{n,\mathrm{PS}}(\mathbf{m})$ and $\eta_{n,\mathrm{CM}}$ capture the separation of the resulting directions. 
In the regimes where the proposed methods provide the largest gains, the dominant issue for MSE-oriented reconstruction is typically geometry collapse, while residual overlap is the cost of posterior sampling and the separation terms control the conditioning of the ZF beams.

\begin{proposition}[Idealized same-projector comparison]
\label{prop:idealized_same_projector}
Fix a tuple $\mathbf{m}$ and a pair $(k,n)$ with $k\neq n$. Assume that
\begin{equation}
\mathbf{\Pi}_{n,\mathrm{CM}} = \mathbf{\Pi}_n(\mathbf{m}).
\label{eq:same_projector_assumption}
\end{equation}
Then $\mathcal{S}_{k,n}^{\mathrm{CM}}(\mathbf{m})=0$.
If, in addition, the cross term is nonnegative, namely
\begin{equation}
\sum_{\ell\neq m_k} p_k(\ell)\, \Re\!\left[ \left( a_{k,n,\ell}^{\mathrm{CM}}(\mathbf{m}) \right)^* b_{k,n,\ell}^{\mathrm{CM}}(\mathbf{m}) \right] \ge 0,
\label{eq:cross_nonnegative_same_projector}
\end{equation}
and
\begin{equation}
\mathcal{C}_{k,n}^{\mathrm{CM}}(\mathbf{m}) >
\left(\frac{\eta_{n,\mathrm{CM}}}{\eta_{n,\mathrm{PS}}(\mathbf{m})} - c_n(\mathbf{m})^2 \right)
\mathcal{R}_{k,n}^{\mathrm{PS}}(\mathbf{m}),
\label{eq:local_condition_same_projector}
\end{equation}
where
$\mathcal{C}_{k,n}^{\mathrm{CM}}(\mathbf{m})
\triangleq
\sum_{\ell\neq m_k}
p_k(\ell)\left|b_{k,n,\ell}^{\mathrm{CM}}(\mathbf{m})\right|^2$,
it follows that
\begin{equation}
\mathcal{I}_{k,n}^{\mathrm{CM}} >
\mathcal{I}_{k,n}^{\mathrm{PS}}(\mathbf{m}),
\label{eq:cm_worse_same_projector}
\end{equation}
that is, conditional-mean averaging is locally worse than posterior sampling.
\end{proposition}

\begin{proof}
See Appendix~\ref{app:proof_idealized_same_projector}.
\end{proof}

Proposition~\ref{prop:idealized_same_projector} should be viewed as an idealized comparison that removes projector mismatch and isolates the effect of averaging within $\hat{\mathbf{u}}_{n,\mathrm{CM}}$. Then, the comparison is governed by how the averaging-induced deviation interacts with the residual overlap left under posterior sampling. In particular, \eqref{eq:cross_nonnegative_same_projector} requires that the contribution from the centroidal deviation $\mathbf{d}_n(\mathbf{m})$ does not cancel the leakage associated with the selected-mode component, so that averaging does not benefit from destructive interaction. Moreover, \eqref{eq:local_condition_same_projector} requires that the resulting centroidal-mismatch term $\mathcal{C}_{k,n}^{\mathrm{CM}}(\mathbf{m})$ be large enough to dominate the residual-overlap term $\mathcal{R}_{k,n}^{\mathrm{PS}}(\mathbf{m})$ even after accounting for the relative projected energies $\eta_{n,\mathrm{CM}}$ and $\eta_{n,\mathrm{PS}}(\mathbf{m})$. 
Thus, posterior sampling is favored when $\mathbf f_n(\mathbf m)$ has small residual overlap with the non-selected posterior modes, while conditional-mean averaging is pulled far enough away from the selected mode that the additional leakage caused by the geometry-collapsing deviation $\mathbf d_n(\mathbf m)$ outweighs any compensating separation advantage.
This analysis helps explain why Euclidean conditional-mean averaging can be poorly matched to ZF interference geometry under multimodal posteriors, and why the proposed flow-based decoders can mitigate this issue by avoiding MSE-oriented centroid collapse.

\section{Numerical Results}
\label{sec:numerical_results}
In this section, we evaluate the proposed flow-matching generative models for CSI feedback in MU-MIMO precoding. 
Both proposed methods use the encoder of the Uniform Enc-Dec (MSE) baseline introduced in Section~\ref{subsec:baseline_methods}. 
The flow-based CSI refiner additionally uses its decoder $D_0(\cdot)$ to obtain the front-end reconstruction.

The proposed flow-based decoders use a four-level 1D U-Net vector-field network with a base channel width of $64$.
The numbers of residual blocks per downsampling and upsampling block, $N_{\mathrm{down}}$ and $N_{\mathrm{up}}$, are set to $1$ and $2$, respectively.
The flow-time embedding uses $\sigma_f=16$.
For inference, the flow ODE is solved by the midpoint method using only $N_{\mathrm{step}}=4$ flow steps, and EMA weights with decay $0.999$ are used.

\subsection{Dataset Specifications}

The simulation dataset is generated using the QuaDRiGa simulator \cite{Jaeckel23} based on the 3GPP 38.901 UMi model \cite{ETSI20}. First, 100{,}000 individual UE channel realizations are generated. Of these, 60{,}000 realizations are used for training. The remaining 40{,}000 realizations are reserved for testing, from which 5{,}000 multiuser channel sets are constructed by randomly selecting a predetermined number $K$ of UE channels for each set.
The carrier frequency is set to 7 GHz, representing the upper midband in FR3. Key simulation parameters, including cell size, base station geometry, channel model, and antenna configurations, are summarized in Table~\ref{tab:dataset_specifications}.

\begin{table}[t]
\centering
\caption{Simulation Dataset Specifications}
\label{tab:dataset_specifications}
\renewcommand{\arraystretch}{1.08}
\setlength{\tabcolsep}{6pt}
\begin{tabular}{ll}
\toprule
\textbf{Parameter} & \textbf{Value} \\
\midrule
Cell type            & Single cell \\
Cell radius          & 100 m \\
BS position          & $(0,0,10)$ m \\
Channel model        & UMi in TR 38.901 \cite{ETSI20} \\
Carrier frequency    & 7 GHz \\
Simulation bandwidth & Narrowband \\
BS antenna setup     & UPA, $4 \times 8$ \\
\bottomrule
\end{tabular}
\end{table}

\subsection{Baseline Methods}
\label{subsec:baseline_methods}

We compare the proposed methods with deterministic CSI feedback baselines, named according to their quantization rule, encoder-decoder architecture, and training objective.

\textbf{Uniform Enc-Dec (MSE)}: This baseline employs a CSINet-style autoencoder \cite{Wen18} featuring uniform quantization, where the non-differentiable quantization operation is approximated via the straight-through estimator (STE) \cite{Bengio13}. Both the encoder and decoder are trained end-to-end to minimize the MSE, serving as the fundamental deterministic reconstruction baseline.

\textbf{Uniform Enc + U-Net Dec (MSE)}: Retaining the same uniform encoder as in the previous baseline, this method replaces the decoder with a U-Net architecture matched to that used in the proposed direct flow decoder. Unlike the proposed method, however, the decoder is trained purely as a deterministic regressor under the MSE objective.

\textbf{Uniform Enc-Dec (Chordal)}: Sharing the same autoencoder architecture and uniform quantization as the Uniform Enc-Dec (MSE) baseline, this scheme is trained with a chordal distance minimization objective. 
Motivated by Theorem~\ref{thm:chordal_bayes_opt}, this sample-based chordal loss can be viewed as a Monte Carlo approximation of the posterior-optimal self-alignment criterion within each feedback cell.

\textbf{$\mu$-law Enc-Dec (MSE)} \cite{Guo20}: This approach substitutes uniform quantization with $\mu$-law quantization while maintaining deterministic end-to-end reconstruction under the MSE objective. Practically corresponding to the CSINet+ framework \cite{Guo20}, it serves as a representative baseline for nonuniform quantization strategies.

\textbf{Per-Lat Quant Enc-Dec (MSE)} \cite{Yin25}: This baseline implements an autoencoder-based latent model equipped with learnable, latent-wise quantization levels. Offering a more sophisticated deterministic quantization design, it provides a critical benchmark to verify that the proposed improvements derive from generative decoding rather than merely from enhanced flexibility in latent quantization.

\begin{figure}[t]
\centering
\subfigure[$K=2$.]{
    \includegraphics
    [width=3.4in]{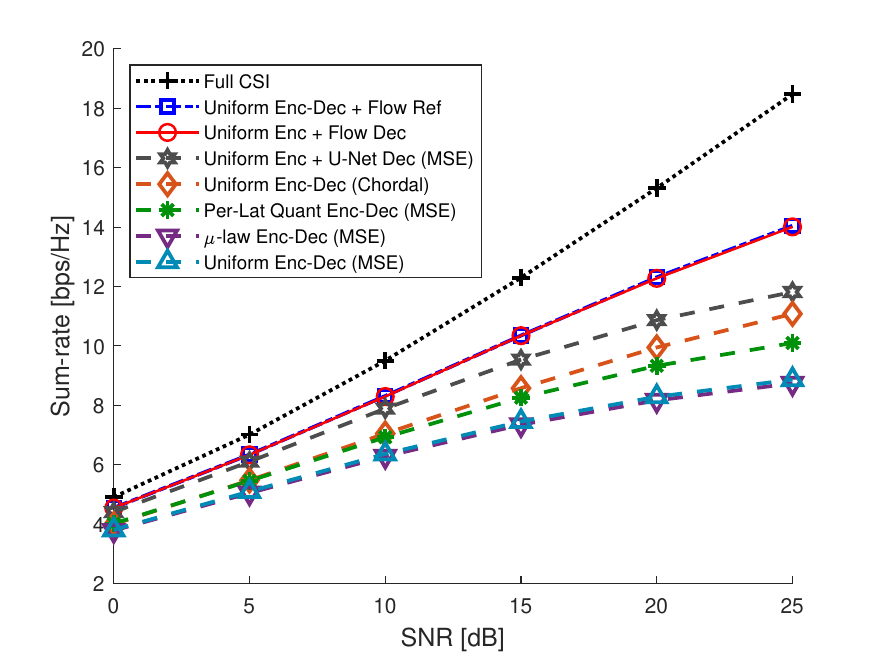}
    \label{fig:Sumrate_lat8_b4_K2}}
\vspace{0.5em}
\subfigure[$K=8$.]{
    \includegraphics
    [width=3.4in]
    {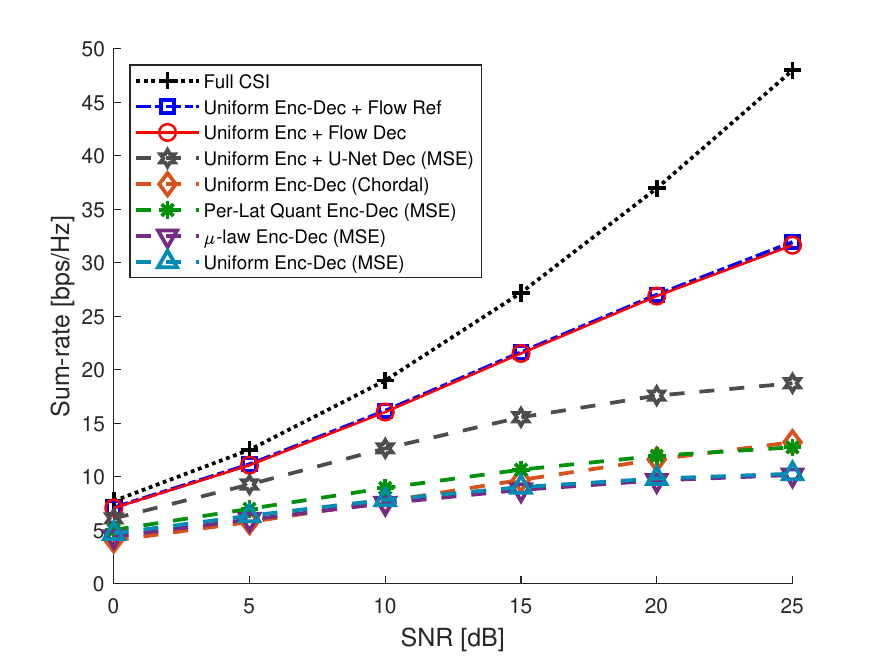}
    \label{fig:Sumrate_lat8_b4_K8}}
\caption{Downlink sum-rate versus SNR for the number of UEs $K\in\{2,8\}$ and latent dimension $L=8$, with $32$ feedback bits per user $(4$ bits per latent element$)$.}
\label{fig:Sumrate_lat8_b4}
\end{figure}

\subsection{Performance Evaluation}
\textbf{The advantage of the proposed methods becomes more pronounced in the interference-limited regime.}
Fig.~\ref{fig:Sumrate_lat8_b4} shows the downlink sum-rate for $K\in\{2,8\}$ with 32 feedback bits per user. The two proposed flow-based methods consistently outperform all deterministic baselines, with the performance gap widening substantially at high SNR, especially for $K=8$. In these interference-limited regimes, ZF precoding is highly sensitive to residual directional mismatches. Notably, even with the same U-Net backbone, the MSE-trained baseline suffers from clear sum-rate saturation. This indicates that the gains of the proposed methods are not due merely to a stronger decoder architecture, but fundamentally to the flow-based generative objective, which better preserves the spatial geometry needed for effective interference suppression.

\begin{figure}[t]
\begin{center}
\includegraphics
[width=3.4in]
{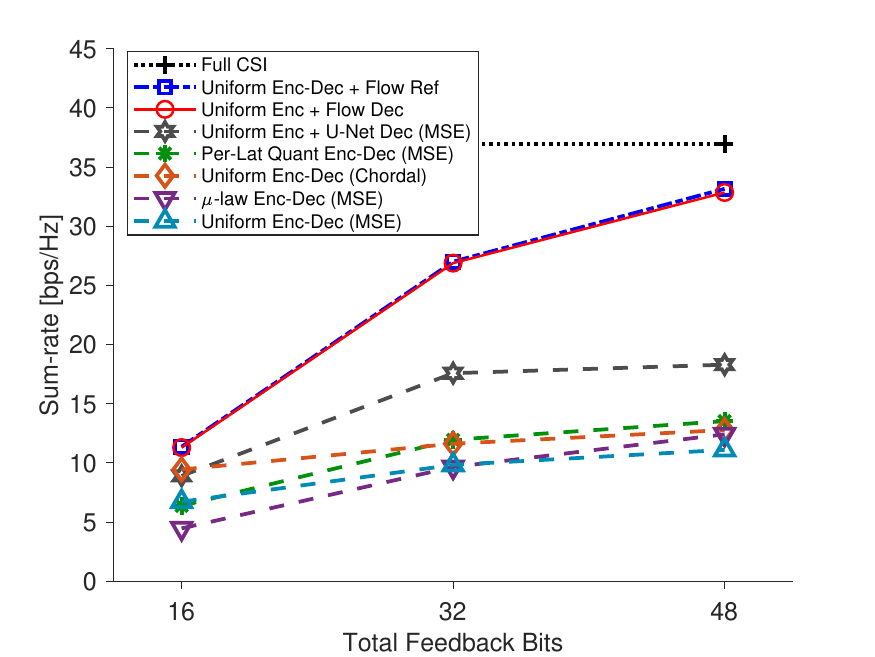}
\end{center}
\caption{Downlink sum-rate versus total feedback bits per user for the number of UEs $K=8$ and latent dimension $L=8$ at SNR $=20$ dB.}
\label{fig:sumrate_vs_bits}
\end{figure}

\textbf{The proposed methods use increased feedback capacity effectively.}
Fig.~\ref{fig:sumrate_vs_bits} illustrates that the proposed flow-based methods achieve the highest sum-rate across all feedback budgets, closely approaching the Full CSI benchmark at $48$ bits. 
Equivalently, both flow-based decoders with $32$ feedback bits per user already outperform the deterministic baselines with $48$ bits per user, indicating at least a $16$-bit-per-user reduction in feedback for the same or higher sum-rate. 
While the flow models exhibit substantial performance scaling as the feedback budget increases from $16$ to $48$ bits, the MSE-based baselines quickly saturate and benefit only marginally.

\begin{figure}[t]
\begin{center}
\includegraphics
[width=3.4in]
{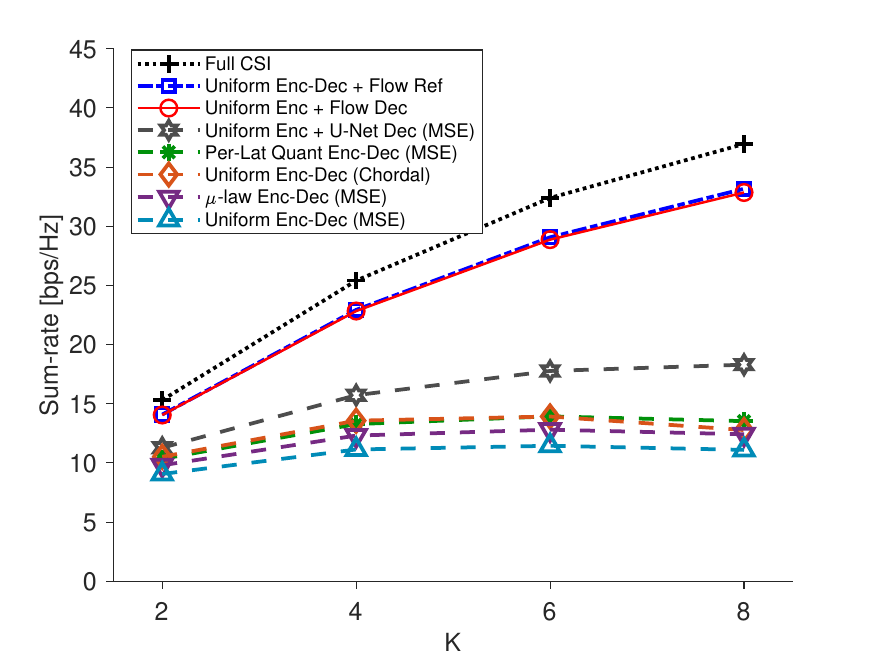}
\end{center}
\caption{Downlink sum-rate versus the number of UEs $K$ with latent dimension $L=8$, $48$ feedback bits per user, and SNR $=20$ dB.}
\label{fig:sumrate_vs_K}
\end{figure}

\textbf{The proposed methods remain highly effective in spatially dense regimes.}
Fig.~\ref{fig:sumrate_vs_K} shows that the performance gap between the proposed methods and the deterministic baselines widens significantly as the number of UEs $K$ increases. 
At $K=8$, the proposed methods achieve about $2.4$ times the sum-rate of the Per-Lat Quant Enc-Dec (MSE) baseline at the same feedback budget, while reaching about $89\%$ of the Full CSI benchmark.
This widening margin reflects the increasing spatial orthogonality constraints in large-$K$ regimes, where ZF precoding becomes especially sensitive to residual directional mismatches.

\begin{table}[t] 
\centering
\caption{NMSE $(dB)$ of CSI feedback methods under feedback budgets of 16, 32, and 48 bits in the UMi setup.}
\label{tab:nmse_bits_comparison}
\renewcommand{\arraystretch}{1.08}
\begin{tabular}{lccc}
\toprule
Method & 16 bits & 32 bits & 48 bits \\
\midrule
\textbf{Uniform Enc-Dec + Flow Ref}   & 2.06  & -1.01 & -4.74 \\
\textbf{Uniform Enc + Flow Dec}  & 1.49  & -0.90 & -4.82 \\
Uniform Enc + U-Net Dec (MSE) & -0.37  & -3.72 & -4.66 \\
Uniform Enc-Dec (Chordal)   & 1.73  & -2.31 & -4.25 \\
Per-Lat Quant Enc-Dec (MSE) & \textbf{-0.58} & \textbf{-3.55} & \textbf{-4.84} \\
$\mu$-law Enc-Dec (MSE)         & 2.25  & -2.09 & -3.42 \\
Uniform Enc-Dec (MSE)       & 1.02  & -0.56 & -2.37 \\
\bottomrule
\end{tabular}
\end{table}

\textbf{Lower NMSE does not necessarily imply higher sum-rate.}
Table~\ref{tab:nmse_bits_comparison} shows that the proposed flow-based methods do not consistently achieve the lowest average NMSE. Deterministic baselines trained under the MSE objective, especially Per-Lat Quant Enc-Dec, achieve a lower average reconstruction error. However, together with Fig.~\ref{fig:Sumrate_lat8_b4}, the table reveals a clear mismatch between average NMSE and MU-MIMO utility, since lower NMSE does not necessarily lead to higher downlink sum-rate. Under finite-rate feedback, MSE-oriented decoding tends to average over multiple plausible channel realizations mapped to the same bits, smoothing the directional structure needed for effective ZF interference suppression. In contrast, the proposed flow-matching methods may accept a slightly higher NMSE while producing reconstructions that better preserve this precoding-relevant geometry, which explains their superior sum-rate performance.

\begin{figure}[t]
\centering
\includegraphics
[width=3.4in]
{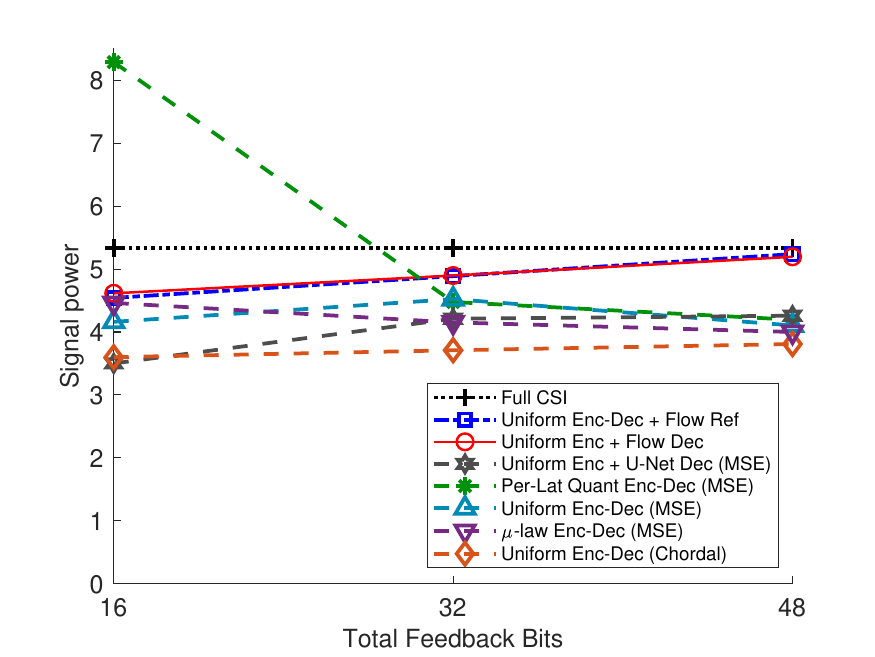}
\caption{Aggregate desired-signal term versus total feedback bits per user for $K=8$, latent dimension $L=8$, and SNR $=20$ dB.}
\label{fig:signal_power_vs_bits}
\end{figure}

\begin{figure}[t]
\centering
\includegraphics
[width=3.4in]
{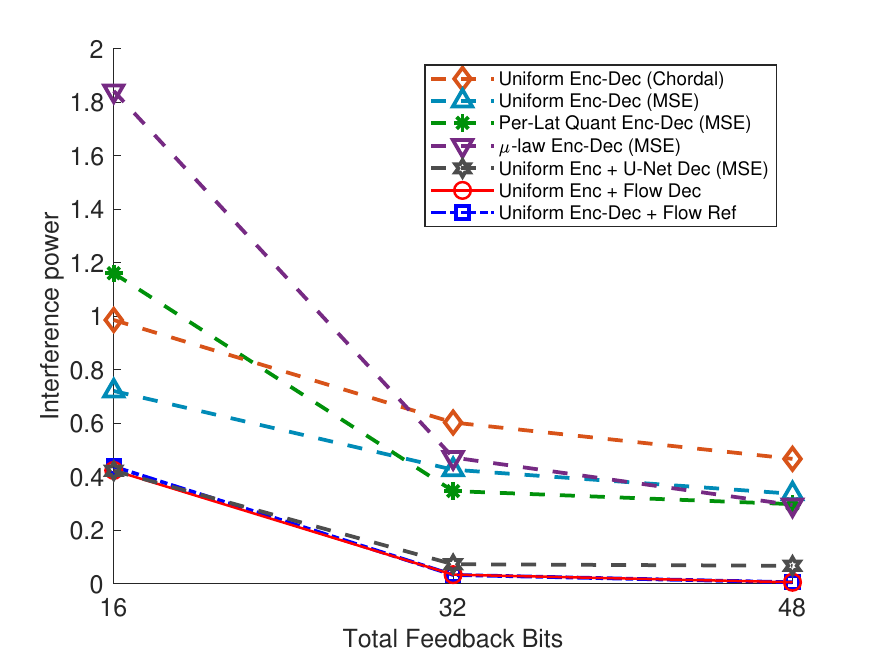}
\caption{Aggregate interference term versus total feedback bits per user for $K=8$, latent dimension $L=8$, and SNR $=20$ dB.}
\label{fig:interference_power_vs_bits}
\end{figure}

\textbf{Desired-signal gain increases with feedback, while the MSE-based baselines saturate.}
Fig.~\ref{fig:signal_power_vs_bits} shows the aggregate desired-signal term obtained by summing the per-user quantity in \eqref{eq:desired_signal} across the $K$ scheduled users, as the total feedback bits increase from $16$ to $48$.
For the proposed flow-based methods, the desired-signal term increases steadily with the feedback budget and approaches Full CSI at $48$ bits, while remaining about $23\%$ higher than that of the strongest MSE-based baseline. The MSE-based baselines exhibit much weaker gains and nearly saturate. 
It is also notable that the Per-Lat Quant baseline yields an aggregate desired-signal term larger than Full CSI at $16$ bits. 
This does not imply a better precoder, since Full CSI ZF is not an upper bound on the desired-signal term alone. 
A beam designed from inaccurate CSI can increase the diagonal projection onto the true channel while violating the true inter-user nulling constraints, leading to much larger residual interference and lower sum-rate.

\textbf{Total interference reduction is the dominant source of sum-rate gain.}
Fig.~\ref{fig:interference_power_vs_bits} shows the corresponding aggregate interference term in \eqref{eq:zf_interference_general}, summed over all user-interference pairs. At $48$ bits, the proposed flow-based methods reduce the total interference to about $0.0055$, whereas the baselines remain much higher, roughly in the range of $0.067$ to $0.47$, corresponding to about a $92$--$99\%$ reduction. Thus, although the proposed methods also improve the desired-signal term, their much larger sum-rate gain is explained primarily by the sharp reduction in total interference. In contrast, the baselines remain limited by somewhat lower aggregate desired-signal power and, more importantly, by much larger interference, indicating that their reconstructed channel directions do not preserve the spatial geometry needed for effective ZF interference suppression.

\begin{figure}[t] 
\centering
\subfigure[$16$ bits per channel. NMSE (dB): Uniform Enc + Flow Dec $=-2.30$, Uniform Enc-Dec + Flow Ref $=-1.53$, Per-Lat Quant Enc-Dec (MSE) $=0.49$, Uniform Enc-Dec (MSE) $=-2.01$.]{
    \includegraphics
    [width=3.4in]
    {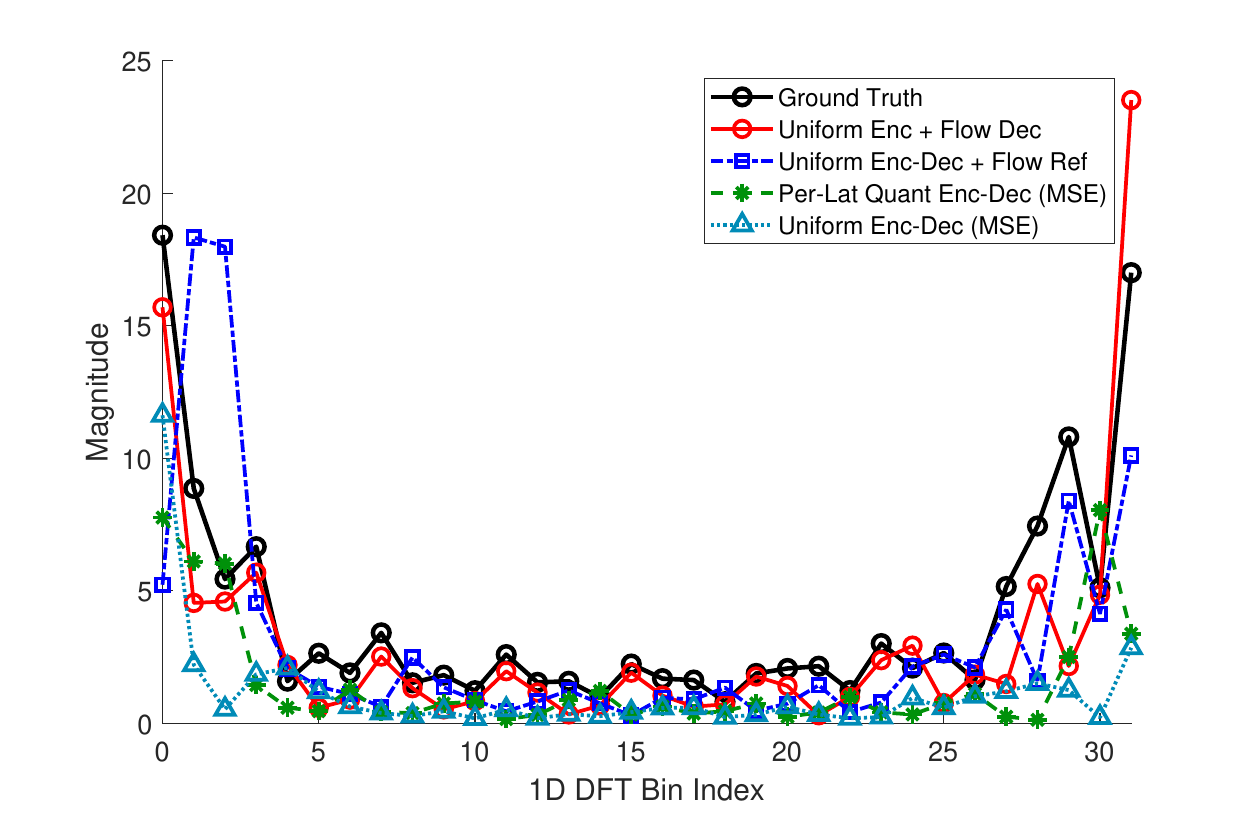}
    \label{fig:fft_mag_example_16bits}}
\vspace{0.5em}
\subfigure[$48$ bits per channel. NMSE (dB): Uniform Enc + Flow Dec $=-32.42$, Uniform Enc-Dec + Flow Ref $=-34.98$, Per-Lat Quant Enc-Dec (MSE) $=-5.96$, Uniform Enc-Dec (MSE) $=-7.99$.]{
    \includegraphics
    [width=3.4in]
    {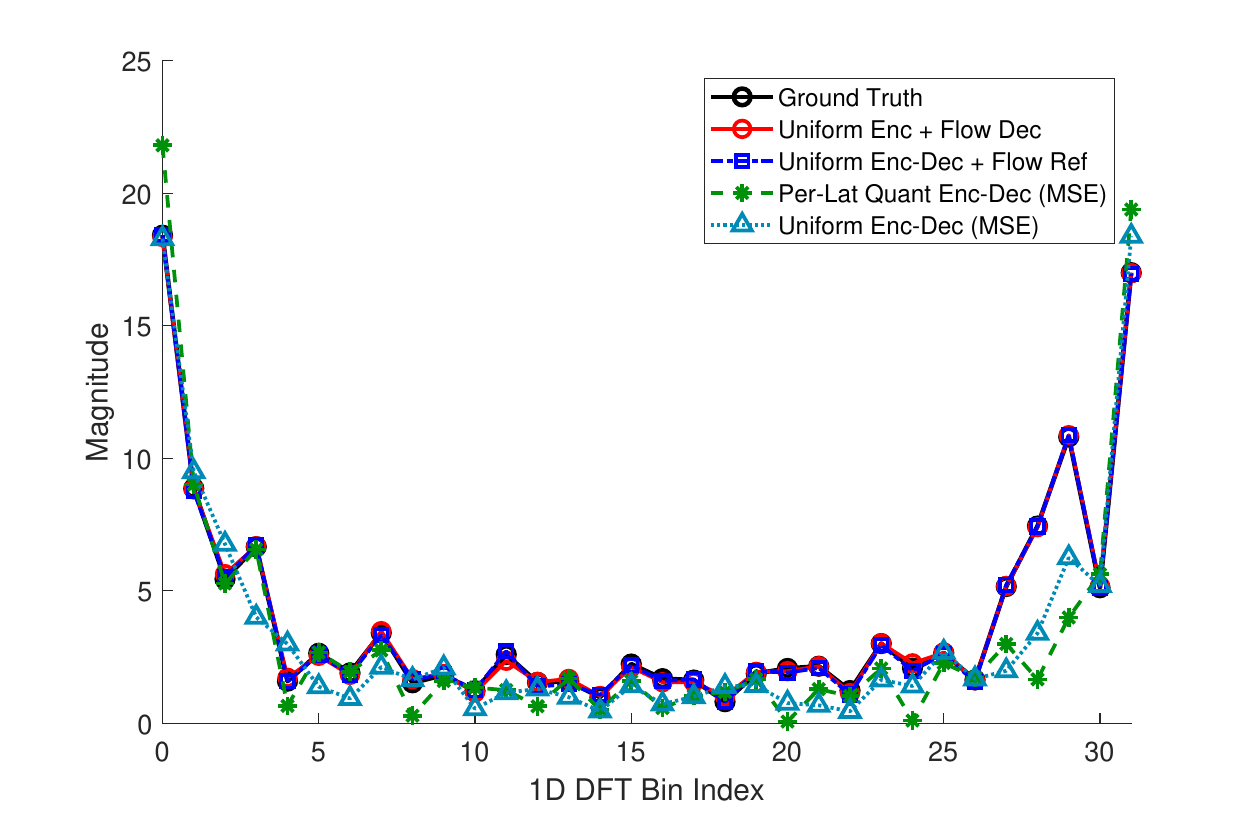}
    \label{fig:fft_mag_example_48bits}}
\caption{Representative 1D DFT-domain magnitude profiles of a channel sample after applying a 1D DFT.}
\label{fig:fft_mag_examples}
\end{figure}

\textbf{Conservative MSE-based reconstruction versus sharper flow-based recovery.}
Fig.~\ref{fig:fft_mag_examples} shows representative 1D discrete Fourier transform (DFT) domain magnitude profiles reconstructed by the different methods. 
At $16$ bits per channel, the MSE-based baselines tend to smooth out sharp dominant peaks and weaken localized DFT-domain components of the ground-truth channel. 
This is consistent with the limited representation capability under such a restrictive feedback budget. 
When the available bits are insufficient to specify the channel sharply, MSE-based training typically favors an averaged reconstruction that minimizes the expected error over the posterior uncertainty, rather than committing to a sharper estimate that might incur a large mismatch penalty. 
In contrast, the proposed flow-based methods preserve the dominant peak structures more effectively and follow the ground-truth profile much more closely, despite some mismatch. 
At $48$ bits per channel, the proposed flow models track the ground truth almost perfectly, whereas the MSE baselines still struggle to reproduce the sharp DFT-domain structures accurately.

\textbf{A similar limitation has been observed in computer vision.}
These visual characteristics help contextualize Table~\ref{tab:nmse_bits_comparison}. Although deterministic MSE training can yield competitive average NMSE, it often does so by smoothing out sharp DFT-domain structures. 
This behavior is analogous to recurring observations in computer vision, where MSE-based objectives often favor blurry reconstructions \cite{Mathieu16, Pathak16, Blau18}. To counter this averaging tendency, earlier work often introduced explicit perceptual objectives. More recently, however, generative models such as diffusion and flow-based models have shown that sharper outputs can emerge from faithful distribution modeling, without the need for an explicit perceptual loss \cite{Ho20, Lipman23}.
Our results suggest a similar lesson for CSI feedback. Even without using the MU-MIMO sum-rate directly as a training loss, the proposed flow-matching decoder can better preserve the spatial channel geometry relevant to ZF precoding by avoiding the conservative averaging behavior of MSE-based reconstruction.

\begin{figure}[t] 
\centering
\subfigure[Example of a \textit{CM-failure} cell with code $(7,9,8,7,8,7,7,7)$, where $\Delta_{\mathrm{CM}}=0.2047$ and $\Delta_{\mathrm{flow}}=0.0851$.]{
    \includegraphics
    [width=3.4in]    {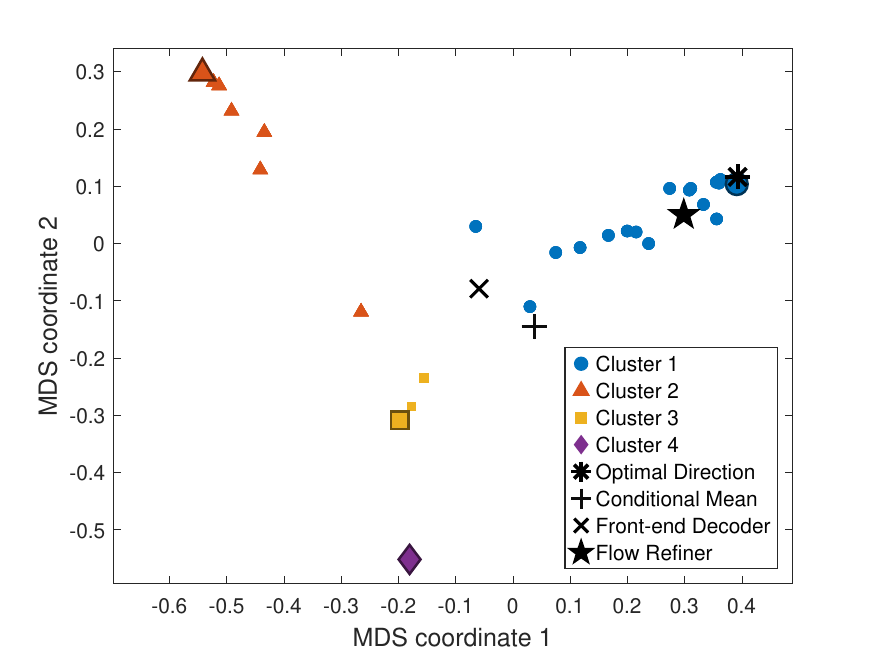}
    \label{fig:poster_shape_flowdec_win}}
\vspace{0.5em}
\subfigure[Example of a \textit{flow-failure} cell with code $(8,8,7,7,8,7,8,8)$, where $\Delta_{\mathrm{CM}}=0.0160$ and $\Delta_{\mathrm{flow}}=0.1449$.]{
    \includegraphics
    [width=3.4in]
    {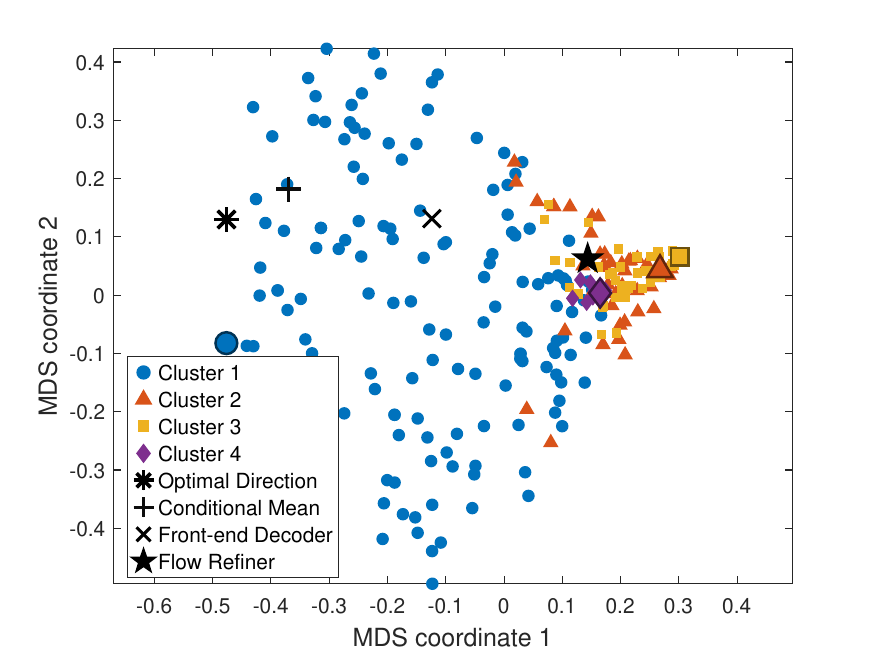}
    \label{fig:poster_shape_flowdec_cmwin}}
\caption{Representative feedback cells for the proposed flow refiner at latent dimension $8$ and $4$ bits per latent element.}
\label{fig:poster_shape_flowdec_examples}
\end{figure}

\textbf{The effectiveness of flow-based decoding depends on the encoder-defined posterior.}

We finally inspect individual feedback cells, defined in \eqref{eqn:def_fb_cell}, to illustrate when posterior-guided decoding helps or hurts. For a feedback vector $\mathbf{b}$, we collect empirical channel directions mapped to the same feedback cell and visualize them using two-dimensional multidimensional scaling (MDS) based on chordal distances, where the MDS coordinates are used only for visualization. 
For a reconstructed direction $\hat{\mathbf{u}}$, we define the gap from the posterior-optimal self-alignment direction as $\Delta(\hat{\mathbf{u}})=\lambda_{\max}-\hat{\mathbf{u}}^H\mathbf{R}(\mathbf{b})\hat{\mathbf{u}}$, where $\lambda_{\max}$ is defined in Theorem~\ref{thm:chordal_bayes_opt}.
Fig.~\ref{fig:poster_shape_flowdec_examples} shows two representative cells for the flow refiner at latent dimension $L=8$ with $4$ bits per latent element. 
In Fig.~\ref{fig:poster_shape_flowdec_win}, the conditional-mean direction lies between visible posterior clusters, while the flow refiner moves closer to the dominant posterior structure, reducing the gap from $\Delta_{\mathrm{CM}}=0.2047$ to $\Delta_{\mathrm{flow}}=0.0851$. 
In Fig.~\ref{fig:poster_shape_flowdec_cmwin}, the opposite behavior occurs: the conditional-mean direction remains closer to the posterior-optimal direction, whereas the flow refiner is pulled toward a weaker mode, increasing the gap from $\Delta_{\mathrm{CM}}=0.0160$ to $\Delta_{\mathrm{flow}}=0.1449$. 
These examples show that flow-based decoding is not guaranteed to improve every feedback cell; its benefit depends on the posterior geometry shaped by the encoder.

\section{Conclusion} \label{sec:conclusion}
This paper showed that effective finite-rate CSI feedback for MU-MIMO precoding depends on preserving posterior geometry rather than relying on MSE-oriented Euclidean averaging. In particular, averaging-based reconstruction can be poorly matched to MU-MIMO precoding, since it may collapse spatially meaningful channel structure into compromise representatives that are less useful for interference suppression. This viewpoint helps explain why lower NMSE does not necessarily result in higher downlink sum-rate.
Important directions for future work include extending the proposed framework and analysis to multi-antenna UEs and exploring UE-side encoder design.

\begin{appendices}

\section{Proof of Proposition~\ref{prop:idealized_same_projector}}
\label{app:proof_idealized_same_projector}

Under \eqref{eq:same_projector_assumption}, we have
$\mathbf{\Pi}_{n,\mathrm{CM}}\mathbf{v}_k(m_k)=\mathbf{0}$,
since
$\mathbf{v}_k(m_k)\in\mathcal{U}_n(\mathbf{m})$.
Hence,
$\mathcal{S}_{k,n}^{\mathrm{CM}}(\mathbf{m})=0$.
Moreover, under the same assumption,
$a_{k,n,\ell}^{\mathrm{CM}}(\mathbf{m}) =
\mathbf{v}_k(\ell)^H \mathbf{\Pi}_n(\mathbf{m}) \mathbf{v}_n(m_n)$,
and therefore
\begin{equation*}
\sum_{\ell\neq m_k}
p_k(\ell) \left| a_{k,n,\ell}^{\mathrm{CM}}(\mathbf{m})\right|^2 =
\mathcal{R}_{k,n}^{\mathrm{PS}}(\mathbf{m}).
\end{equation*}

Substituting
$\mathcal{S}_{k,n}^{\mathrm{CM}}(\mathbf{m})=0$
into \eqref{eq:local_cm_exact} gives
\begin{align*}
\mathcal{I}_{k,n}^{\mathrm{CM}}
&= \frac{1}{\eta_{n,\mathrm{CM}}} \sum_{\ell\neq m_k}
p_k(\ell) \left| c_n(\mathbf{m})\,a_{k,n,\ell}^{\mathrm{CM}}(\mathbf{m}) + b_{k,n,\ell}^{\mathrm{CM}}(\mathbf{m}) \right|^2 \\
&= \frac{1}{\eta_{n,\mathrm{CM}}} \Bigg( c_n(\mathbf{m})^2
\sum_{\ell\neq m_k} p_k(\ell) \left| a_{k,n,\ell}^{\mathrm{CM}}(\mathbf{m}) \right|^2 +
\mathcal{C}_{k,n}^{\mathrm{CM}}(\mathbf{m})\\
&\qquad + 2c_n(\mathbf{m}) \sum_{\ell\neq m_k} p_k(\ell)
\Re\!\left[ \left( a_{k,n,\ell}^{\mathrm{CM}}(\mathbf{m}) \right)^* b_{k,n,\ell}^{\mathrm{CM}}(\mathbf{m}) \right] \Bigg).
\end{align*}
By \eqref{eq:cross_nonnegative_same_projector} and the identity above,
\begin{equation*}
\mathcal{I}_{k,n}^{\mathrm{CM}}
\ge
\frac{ c_n(\mathbf{m})^2
\mathcal{R}_{k,n}^{\mathrm{PS}}(\mathbf{m}) +
\mathcal{C}_{k,n}^{\mathrm{CM}}(\mathbf{m})}{\eta_{n,\mathrm{CM}}}.
\end{equation*}

On the other hand, by \eqref{eq:local_ps_exact},
$
\mathcal{I}_{k,n}^{\mathrm{PS}}(\mathbf{m}) = \frac{\mathcal{R}_{k,n}^{\mathrm{PS}}(\mathbf{m})}{
\eta_{n,\mathrm{PS}}(\mathbf{m})}.
$
Hence, if \eqref{eq:local_condition_same_projector} holds, then
\begin{equation*}
c_n(\mathbf{m})^2 \mathcal{R}_{k,n}^{\mathrm{PS}}(\mathbf{m}) +
\mathcal{C}_{k,n}^{\mathrm{CM}}(\mathbf{m}) >
\frac{\eta_{n,\mathrm{CM}}}{\eta_{n,\mathrm{PS}}(\mathbf{m})} \mathcal{R}_{k,n}^{\mathrm{PS}}(\mathbf{m}),
\end{equation*}
which implies
$\mathcal{I}_{k,n}^{\mathrm{CM}} >
\mathcal{I}_{k,n}^{\mathrm{PS}}(\mathbf{m})$.

\end{appendices}

\bibliographystyle{ieeetr}
\begingroup
\bibliography{JuseongREF}

@ARTICLE{Sohrabi21,
  author={Sohrabi, Foad and Attiah, Kareem M. and Yu, Wei},
  journal={IEEE Trans. Wireless Commun.}, 
  title={Deep Learning for Distributed Channel Feedback and Multiuser Precoding in {FDD} Massive {MIMO}}, 
  month={Jul.},
  year={2021},
  volume={20},
  number={7},
  pages={4044-4057},
  doi={10.1109/TWC.2021.3055202}}

@inproceedings{
  Welker25,
  title={{FlowDec}: A flow-based full-band general audio codec with high perceptual quality},
  author={Simon Welker and Matthew Le and Ricky T. Q. Chen and Wei-Ning Hsu and Timo Gerkmann and Alexander Richard and Yi-Chiao Wu},
  booktitle={Int. Conf. Learn. Represent. (ICLR)},
  year={2025},
  month={Apr.},
}

@inproceedings{Lipman23,
  author    = {Yaron Lipman and Ricky T. Q. Chen and Heli Ben-Hamu and Maximilian Nickel and Matthew Le},
  title     = {Flow Matching for Generative Modeling},
  booktitle = {Proc. Int. Conf. Learn. Represent., (ICLR)},
  year      = {2023},
  month     = may,
}

@inproceedings{Song21,
  author    = {Yang Song and Jascha Sohl-Dickstein and Diederik P. Kingma and Abhishek Kumar and Stefano Ermon and Ben Poole},
  title     = {Score-Based Generative Modeling through Stochastic Differential Equations},
  booktitle = {Proc. Int. Conf. Learn. Represent. (ICLR)},
  year      = {2021}
}

@inproceedings{Ho20,
  title={Denoising diffusion probabilistic models},
  author={Ho, Jonathan and Jain, Ajay and Abbeel, Pieter},
  booktitle={Proc. Adv. Neural Inf. Process. Syst. (NeurIPS)},
  volume={33},
  pages={6840--6851},
  year={2020}
}

@book{Kay93,
  author    = {Steven M. Kay},
  title     = {Fundamentals of Statistical Signal Processing: Estimation Theory},
  publisher = {Prentice-Hall},
  address   = {Upper Saddle River, NJ, USA},
  year      = {1993},
  isbn      = {978-0-13-345711-7}
}

@ARTICLE{Yin25,
  author={Yin, Manru and Han, Shengqian and Yang, Chenyang},
  journal={IEEE Wireless Commun. Lett.}, 
  title={Quantization Design for Deep Learning-Based {CSI} Feedback}, 
  year={2025},
  month={Aug.},
  volume={14},
  number={8},
  pages={2411-2415},
  doi={10.1109/LWC.2025.3570965}}

@misc{ETSI20,
  author = {{ETSI}},
  title = {Study on channel model for frequencies from 0.5 to 100 {GHz}},
  howpublished = {ETSI TR 138 901, V16.1.0},
  month = nov,
  year = {2020},
  url = {https://www.etsi.org/deliver/etsi_tr/138900_138999/138901/16.01.00_60/tr_138901v160100p.pdf}
}

@ARTICLE{Guo20,
  author={Guo, Jiajia and Wen, Chao-Kai and Jin, Shi and Li, Geoffrey Ye},
  journal={IEEE Trans. Wireless Commun.}, 
  title={Convolutional Neural Network-Based Multiple-Rate Compressive Sensing for Massive {MIMO CSI} Feedback: Design, Simulation, and Analysis}, 
  year={2020},
  month={Apr.},
  volume={19},
  number={4},
  pages={2827-2840},
  doi={10.1109/TWC.2020.2968430}}

@ARTICLE{Wen18,
  author={Wen, Chao-Kai and Shih, Wan-Ting and Jin, Shi},
  journal={IEEE Wireless Commun. Lett.}, 
  title={Deep Learning for Massive {MIMO CSI} Feedback}, 
  year={2018},
  month={Oct.},
  volume={7},
  number={5},
  pages={748-751},
  doi={10.1109/LWC.2018.2818160}}

@misc{Bengio13,
  author       = {Yoshua Bengio and
                  Nicholas L{\'{e}}onard and
                  Aaron C. Courville},
  title        = {Estimating or Propagating Gradients Through Stochastic Neurons for Conditional Computation},
  year={2013},
  eprint={1308.3432},
  archivePrefix={arXiv},
  note          = {[Online]. Available: \url{http://arxiv.org/abs/1308.3432}},}

@TechReport{Jaeckel23,
  author={Jaeckel, Stephan and Raschkowski, Leszek and Börner, Kai and Thiele, Lars},
  title={{QuaDRiGa} — {Quasi} deterministic radio channel generator, user manual and documentation}, 
  institution = {Fraunhofer Heinrich Hertz Institute},
  year      = {2023},
  number    = {v2.8.1},
}

@ARTICLE{Guo22,
  author={Guo, Jiajia and Wen, Chao-Kai and Jin, Shi and Li, Geoffrey Ye},
  journal={IEEE Trans. Commun.}, 
  title={Overview of Deep Learning-Based {CSI} Feedback in Massive {MIMO} Systems}, 
  year={2022},
  month={Dec.},
  volume={70},
  number={12},
  pages={8017-8045},
  doi={10.1109/TCOMM.2022.3217777}}

@ARTICLE{Guo25,
  author={Guo, Yiran and Chen, Wei and Sun, Feifei and Cheng, Jiaming and Matthaiou, Michail and Ai, Bo},
  journal={IEEE Commun. Mag.}, 
  title={Deep Learning for {CSI} Feedback: One-Sided Model and Joint Multi-Module Learning Perspectives}, 
  year={2025},
  month={Jul.},
  volume={63},
  number={7},
  pages={90-97},
  doi={10.1109/MCOM.001.2400285}}

@ARTICLE{Xingqin25,
  author={Lin, Xingqin},
  journal={IEEE Commun. Standards Mag.}, 
  title={A Tale of Two Mobile Generations: {5G-Advanced} and {6G} in {3GPP} {Release} 20}, 
  year={2026},
  volume={},
  number={},
  pages={},
  note={\emph{to appear}},
  doi={10.1109/MCOMSTD.2025.3613219}}

@INPROCEEDINGS{Kim25,
  author={Kim, Heasung and Lee, Taekyun and Kim, Hyeji and De Veciana, Gustavo and Arfaoui, Mohamed Amine and Koc, Asil and Pietraski, Phil and Zhang, Guodong and Kaewell, John},
  booktitle={Proc. IEEE Int. Conf. Commun. (ICC)}, 
  title={Generative Diffusion Model-Based Compression of {MIMO CSI}}, 
  year={2025},
  volume={},
  number={},
  pages={6323-6328},
  doi={10.1109/ICC52391.2025.11161629}}

@ARTICLE{Ankireddy26,
  author={Ankireddy, Sravan Kumar and Kim, Heasung and Cho, Joonyoung and Kim, Hyeji},
  journal={IEEE J. Sel. Areas Commun.}, 
  title={Residual Diffusion Models for Variable-Rate Joint Source–Channel Coding of {MIMO CSI}}, 
  year={2026},
  volume={44},
  number={},
  pages={3620-3633},
  doi={10.1109/JSAC.2026.3666741}}

@ARTICLE{Park25,
  author={Park, Juseong and Sohrabi, Foad and Ghosh, Amitava and Andrews, Jeffrey G.},
  journal={IEEE Trans. Wireless Commun.}, 
  title={End-to-End Deep Learning for {TDD MIMO} Systems in the {6G} Upper Midbands}, 
  year={2025},
  month={Mar.},
  volume={24},
  number={3},
  pages={2110-2125},
  doi={10.1109/TWC.2024.3516633}}

@ARTICLE{Jang22,
  author={Jang, Jeonghyeon and Lee, Hoon and Kim, Il-Min and Lee, Inkyu},
  journal={IEEE Trans. Commun.}, 
  title={Deep Learning for Multi-User {MIMO} Systems: Joint Design of Pilot, Limited Feedback, and Precoding}, 
  month={Nov.},
  year={2022},
  volume={70},
  number={11},
  pages={7279-7293},
  doi={10.1109/TCOMM.2022.3209887}}

@ARTICLE{Lu20_uniform,
  author={Lu, Chao and Xu, Wei and Jin, Shi and Wang, Kezhi},
  journal={IEEE Wireless Commun. Lett.}, 
  title={Bit-Level Optimized Neural Network for Multi-Antenna Channel Quantization}, 
  year={2020},
  month={Jan.},
  volume={9},
  number={1},
  pages={87-90},
  doi={10.1109/LWC.2019.2942908}}

@ARTICLE{Rizzello23,
  author={Rizzello, Valentina and Nerini, Matteo and Joham, Michael and Clerckx, Bruno and Utschick, Wolfgang},
  journal={IEEE Wireless Commun. Lett.}, 
  title={User-Driven Adaptive {CSI} Feedback With Ordered Vector Quantization}, 
  year={2023},
  month={Nov.},
  volume={12},
  number={11},
  pages={1956-1960},
  doi={10.1109/LWC.2023.3301992}}

@inproceedings{Mathieu16,
  author={Mathieu, Michael and Couprie, Camille and LeCun, Yann},
  title={Deep multi-scale video prediction beyond mean square error},
  booktitle={Proc. Int. Conf. Learn. Represent. (ICLR)},
  year={2016},
  month={May},
  pages={1--14}}

@InProceedings{Pathak16,
author = {Pathak, Deepak and Krahenbuhl, Philipp and Donahue, Jeff and Darrell, Trevor and Efros, Alexei A.},
title = {Context Encoders: Feature Learning by Inpainting},
booktitle={Proc. IEEE Conf. Comput. Vis. Pattern Recognit. (CVPR)},
month = {Jun.},
year = {2016},
pages={2536--2544}}

@inproceedings{Blau18,
  title={The perception-distortion tradeoff},
  author={Blau, Yochai and Michaeli, Tomer},
  booktitle={Proc. IEEE/CVF Conf. Comput. Vis. Pattern Recognit. (CVPR)},
  month={Jun.},
  year={2018},
  pages={6228--6237}
}

@ARTICLE{Jindal06,
  author={Jindal, N.},
  journal={IEEE Trans. Inf. Theory}, 
  title={{MIMO} Broadcast Channels With Finite-Rate Feedback}, 
  year={2006},
  month={Nov.},
  volume={52},
  number={11},
  pages={5045-5060},
  doi={10.1109/TIT.2006.883550}}

@ARTICLE{Love03,
  author={Love, D.J. and Heath, R.W. and Strohmer, T.},
  journal={IEEE Trans. Inf. Theory}, 
  title={Grassmannian beamforming for multiple-input multiple-output wireless systems}, 
  year={2003},
  month={Oct.},
  volume={49},
  number={10},
  pages={2735-2747},
  doi={10.1109/TIT.2003.817466}}

@ARTICLE{Ravindran08,
  author={Ravindran, Niranjay and Jindal, Nihar},
  journal={IEEE J. Sel. Areas Commun.}, 
  title={Limited feedback-based block diagonalization for the {MIMO} broadcast channel}, 
  year={2008},
  month={Oct.},
  volume={26},
  number={8},
  pages={1473-1482},
  doi={10.1109/JSAC.2008.081013}}

@book{Heath18, title = {Foundations of MIMO communication}, author =  {R. W. Heath and A. Lozano}, isbn = {}, series = {}, year = {2018}, address = { Cambridge, U.K.},  publisher = {Cambridge Univ. Press}}

@ARTICLE{Love08,
  author={Love, David J. and Heath, Robert W. and N. Lau, Vincent K. and Gesbert, David and Rao, Bhaskar D. and Andrews, Matthew},
  journal={IEEE J. Sel. Areas Commun.}, 
  title={An overview of limited feedback in wireless communication systems}, 
  year={2008},
  month={Oct.},
  volume={26},
  number={8},
  pages={1341-1365},
  doi={10.1109/JSAC.2008.081002}}

@ARTICLE{Zhu26,
  author={Zhu, Lin and Zhu, Weifeng and Zhang, Shuowen and Cui, Shuguang and Liu, Liang},
  journal={IEEE Trans. Wireless Commun.}, 
  title={Scalable Transceiver Design for Multi-User Communication in {FDD} Massive {MIMO} Systems via Deep Learning}, 
  year={2026},
  volume={25},
  number={},
  pages={7682-7697},
  doi={10.1109/TWC.2025.3633144}}

@ARTICLE{Carpi26,
  author={Carpi, Fabrizio and Venkatesan, Sivarama and Du, Jinfeng and Viswanathan, Harish and Garg, Siddharth and Erkip, Elza},
  journal={IEEE Trans. Wireless Commun.}, 
  title={Learned Precoding-Oriented {CSI} Feedback in Multi-Cell Multi-User {MIMO} Systems}, 
  year={2026},
  volume={25},
  number={},
  pages={2359-2372},
  doi={10.1109/TWC.2025.3596140}}

@ARTICLE{Goldsmith03,
  author={Goldsmith, A. and Jafar, S.A. and Jindal, N. and Vishwanath, S.},
  journal={IEEE J. Sel. Areas Commun.}, 
  title={Capacity limits of {MIMO} channels}, 
  year={2003},
  month={Jun.},
  volume={21},
  number={5},
  pages={684-702},
  doi={10.1109/JSAC.2003.810294}}

@ARTICLE{Andrews25,
  author={Andrews, Jeffrey G. and Humphreys, Todd E. and Ji, Tingfang},
  journal={IEEE BITS Inf. Theory Mag.}, 
  title={{6G} Takes Shape}, 
  year={2025},
  month={Mar.},
  volume={4},
  number={1},
  pages={2-24},
  doi={10.1109/MBITS.2024.3504521}}

@ARTICLE{Fan26,
  author={Fan, Dayu and Meng, Rui and Xu, Xiaodong and Liu, Yiming and Nan, Guoshun and Feng, Chenyuan and Han, Shujun and Gao, Song and Xu, Bingxuan and Niyato, Dusit and Quek, Tony Q. S. and Zhang, Ping},
  journal={IEEE Commun. Surv. Tutor.}, 
  title={Generative Diffusion Models for Wireless Networks: Fundamental, Architecture, and State-of-the-Art}, 
  year={2026},
  volume={28},
  number={},
  pages={5632-5677},
  doi={10.1109/COMST.2026.3671110}}

@ARTICLE{Fesl24,
  author={Fesl, Benedikt and Baur, Michael and Strasser, Florian and Joham, Michael and Utschick, Wolfgang},
  journal={IEEE Wireless Commun. Lett.}, 
  title={Diffusion-Based Generative Prior for Low-Complexity {MIMO} Channel Estimation}, 
  year={2024},
  month={Dec.},
  volume={13},
  number={12},
  pages={3493-3497},
  doi={10.1109/LWC.2024.3474570}}

@ARTICLE{Chen25,
  author={Chen, Zhixiong and Shin, Hyundong and Nallanathan, Arumugam},
  journal={IEEE Trans. Commun.}, 
  title={Generative Diffusion Model-Based Variational Inference for {MIMO} Channel Estimation}, 
  year={2025},
  month={Oct.},
  volume={73},
  number={10},
  pages={9254-9269},
  doi={10.1109/TCOMM.2025.3556753}}

@misc{Chen26,
  author        = {Xi Chen and
                   Homa Esfahanizadeh and
                   Foad Sohrabi},
  title         = {Precoding-Oriented {CSI} Feedback Design with Mutual Information Regularized {VQ-VAE}},
  year          = {2026},
  eprint        = {2602.02508},
  archivePrefix = {arXiv},
  note          = {[Online]. Available: \url{https://arxiv.org/abs/2602.02508}},
}
\endgroup

\end{document}